\documentclass[12pt]{article}
\usepackage{nips11submit_e,times}
\RequirePackage{amsbsy,amsmath,amssymb,amsthm}
\RequirePackage{graphicx}
\RequirePackage{subfigure}

\newcommand{\reals}{\mathbb{R}}

\DeclareMathOperator*{\Tr}{Tr}
\DeclareMathOperator*{\diag}{diag}
\DeclareMathOperator*{\rank}{rank}

\newcommand{\prob}{\mathbb{P}}
\newcommand{\E}{\mathbb{E}}
\DeclareMathOperator*{\RSS}{RSS}

\theoremstyle{plain}
\newtheorem{theorem}{Theorem}[section]

\newtheorem{lemma}[theorem]{Lemma}
\newtheorem{proposition}[theorem]{Proposition}

\title{
  Regularized Laplacian Estimation
  and \\ Fast Eigenvector Approximation
}
\author{
  Patrick O.~Perry \\
  Information, Operations, and Management Sciences \\
  NYU Stern School of Business \\
  New York, NY 10012 \\
  \texttt{pperry@stern.nyu.edu} 
  \And
  Michael W.~Mahoney \\
  Department of Mathematics \\
  Stanford University \\
  Stanford, CA 94305  \\
  \texttt{mmahoney@cs.stanford.edu}
}

\nipsfinalcopy 

\begin{document}

\maketitle

\vspace{-8mm}
\begin{abstract}
Recently, Mahoney and Orecchia demonstrated that popular diffusion-based 
procedures to compute a quick \emph{approximation} to the first nontrivial 
eigenvector of a data graph Laplacian \emph{exactly} solve certain 
regularized Semi-Definite Programs (SDPs). 
In this paper, we extend that result by providing a statistical 
interpretation of their approximation procedure.
Our interpretation will be analogous to the manner in which
$\ell_2$-regularized or $\ell_1$-regularized $\ell_2$-regression (often
called Ridge regression and Lasso regression, respectively) can be
interpreted in terms of a Gaussian prior or a Laplace prior, respectively,
on the coefficient vector of the regression problem.
Our framework will imply that the solutions to the Mahoney-Orecchia 
regularized SDP can be interpreted as regularized estimates of the 
pseudoinverse of the graph Laplacian.
Conversely, it will imply that the solution to this regularized estimation 
problem can be computed very quickly by running, \emph{e.g.}, the fast 
diffusion-based PageRank procedure for computing an approximation to the 
first nontrivial eigenvector of the graph Laplacian.
Empirical results are also provided to illustrate the manner in which 
approximate eigenvector computation \emph{implicitly} performs statistical 
regularization, relative to running the corresponding exact algorithm.
\end{abstract}

\vspace{-5mm}
\section{Introduction}
\label{sxn:intro}

Approximation algorithms and heuristic approximations are commonly used to
speed up the running time of algorithms in machine learning and data
analysis.
In some cases, the outputs of these approximate procedures are ``better''
than the output of the more expensive exact algorithms, in the sense that
they lead to more robust results or more useful results for the downstream
practitioner.
Recently, Mahoney and Orecchia formalized these ideas in the context of
computing the first nontrivial eigenvector of a graph
Laplacian~\cite{MO11-implementing}.
Recall that, given a graph $G$ on $n$ nodes or equivalently its $n \times n$
Laplacian matrix $L$, the top nontrivial eigenvector of the Laplacian
\emph{exactly} optimizes the Rayleigh quotient, subject to the usual
constraints.
This optimization problem can equivalently be expressed as a vector
optimization program with the objective function $f(x) = x^TLx$,
where $x$ is an $n$-dimensional vector, or as a Semi-Definite Program (SDP)
with objective function $F(X)=\mathrm{Tr}(L X)$, where $X$ is an $n \times n$
symmetric positive semi-definite matrix.
This first nontrivial vector is, of course, of widespread interest in
applications due to its usefulness for graph partitioning, image
segmentation, data clustering, semi-supervised learning, etc.~\cite{spielman96_spectral,guatterymiller98,ShiMalik00_NCut,BN03,Joa03,LLDM09_communities_IM}.

In this context, Mahoney and Orecchia asked the question: do popular
diffusion-based procedures---such as running the Heat Kernel or performing a
Lazy Random Walk or computing the PageRank function---to compute a quick
\emph{approximation} to the first nontrivial eigenvector of $L$ solve some
other regularized version of the Rayleigh quotient objective function
\emph{exactly}?
Understanding this algorithmic-statistical tradeoff is clearly of interest
if one is interested in very large-scale applications, where performing
statistical analysis to derive an objective and then calling a black box
solver to optimize that objective exactly might be too expensive.
Mahoney and Orecchia answered the above question in the affirmative, with
the interesting twist that the regularization is on the SDP formulation
rather than the usual vector optimization problem.
That is, these three diffusion-based procedures exactly optimize a 
regularized SDP with objective function $F(X)+ \tfrac{1}{\eta} G(X)$, for 
some regularization function $G(\cdot)$ to be described below, subject to 
the usual constraints.

\vspace{-1mm}
In this paper, we extend the Mahoney-Orecchia result by providing a
statistical interpretation of their approximation procedure.
Our interpretation will be analogous to the manner in which
$\ell_2$-regularized or $\ell_1$-regularized $\ell_2$-regression (often
called Ridge regression and Lasso regression, respectively) can be
interpreted in terms of a Gaussian prior or a Laplace prior, respectively,
on the coefficient vector of the regression problem.
In more detail, we will set up a sampling model, whereby the graph Laplacian
is interpreted as an observation from a random process; we will posit the
existence of a ``population Laplacian'' driving the random process; and we
will then define an estimation problem: find the inverse of the population
Laplacian.
We will show that the maximum a posteriori probability (MAP) estimate of 
the inverse of the population Laplacian leads to a regularized SDP, where 
the objective function $F(X)=\mathrm{Tr}(L X)$ and where the role of the 
penalty function $G(\cdot)$ is to encode prior assumptions about the 
population Laplacian.
In addition, we will show that when $G(\cdot)$ is the log-determinant 
function then the MAP estimate leads to the Mahoney-Orecchia regularized 
SDP corresponding to running the PageRank heuristic.
Said another way, the solutions to the Mahoney-Orecchia regularized SDP can 
be interpreted as regularized estimates of the pseudoinverse of the graph 
Laplacian.
Moreover, by Mahoney and Orecchia's main result, the solution to this 
regularized SDP can be computed very quickly---rather than solving 
the SDP with a black-box solver and rather computing explicitly the 
pseudoinverse of the Laplacian, one can simply run the fast diffusion-based 
PageRank heuristic for computing an approximation to the first nontrivial 
eigenvector of the Laplacian $L$.

\vspace{-1mm}
The next section describes some background. 
Section~\ref{snx:framework} then describes a statistical framework for graph
estimation; and Section~\ref{sxn:priors} describes prior assumptions that
can be made on the population Laplacian.
These two sections will shed light on the computational implications
associated with these prior assumptions; but more importantly they will shed
light on the implicit prior assumptions associated with making certain
decisions to speed up computations.
Then, Section~\ref{sxn:empirical} will provide an empirical evaluation,
and 
Section~\ref{sxn:conc} will provide a brief conclusion.
Additional discussion is available in the Appendix.

\vspace{-4mm}
\section{Background on Laplacians and diffusion-based procedures}
\label{S:introduction}
\vspace{-3mm}

A weighted symmetric graph $G$ is defined by a vertex set 
$V = \{ 1, \dotsc, n \}$, an edge set $E \subset V \times V$, 
and a weight function $w : E \to \reals_+$, where $w$ is assumed to 
be symmetric 
(\emph{i.e.}, $w(u,v) = w(v,u)$).  
In this case, one can construct a matrix, $L_0 \in \reals^{V \times V}$, 
called the combinatorial Laplacian of $G$:
\[
  L_0(u,v)
  =
  \begin{cases}
    - w(u,v) & \text{when $u \neq v$,} \\
    d(u) - w(u,u) & \text{otherwise,}
  \end{cases}
\]
where $d(u) = \sum_{v} w(u,v)$ is called the degree of $u$.
By construction, $L_0$ is positive semidefinite.  
Note that the all-ones vector, often denoted $1$, is an eigenvector of 
$L_0$ with eigenvalue zero, \emph{i.e.}, $L 1 = 0$.
For this reason, $1$ is often called trivial eigenvector of $L_0$.  
Letting $D$ be a diagonal matrix with $D(u,u) = d(u)$, one can also define 
a normalized version of the Laplacian: $L = D^{-1/2} L_0 D^{-1/2}$.  
Unless explicitly stated otherwise, when we refer to the Laplacian of a 
graph, we will mean the normalized Laplacian.

\vspace{-1mm}
In many situations, \emph{e.g.}, to perform spectral graph partitioning, 
one is interested in computing the first \emph{nontrivial} eigenvector of 
a Laplacian.
Typically, this vector is computed ``exactly'' by calling a black-box 
solver; but it could also be approximated with an iteration-based method 
(such as the Power Method or Lanczos Method) or by running a random 
walk-based or diffusion-based method to the asymptotic state.
These random walk-based or diffusion-based methods assign positive and 
negative ``charge'' to the nodes, and then they let the distribution of 
charge evolve according to dynamics derived from the graph structure.
Three canonical evolution dynamics are the following:
\begin{description}
  \item[Heat Kernel.]
    Here, the charge evolves according to the 
    heat equation
    $\frac{\partial H_t}{\partial t} = - L H_t$.
    Thus, the vector of charges evolves as
$     H_t = \exp ( -tL )  = \sum_{k=0}^{\infty} \frac{(-t)^k}{k!}L^k  $,
    where $t \ge 0$ is a time parameter, times an input seed distribution vector.
  \item[PageRank.]
    Here, the charge at a node evolves by either moving to a neighbor of 
    the current node or teleporting to a random node.
    More formally, the vector of charges evolves as 
    \begin{equation}
    \label{eqn:page-rank}
    R_{\gamma} = \gamma \left(I-\left(1-\gamma \right)M \right)^{-1}   ,
    \end{equation}
    where $M$ is the natural random walk transition matrix associated 
    with the graph and
    where $\gamma \in (0,1)$ is the so-called teleportation parameter,
    times an input seed vector.
  \item[Lazy Random Walk.]
    Here, the charge either stays at the current node or moves to a neighbor.
    Thus, if $M$ is the natural random walk transition matrix associated 
    with the graph, then the vector of charges evolves as some power of 
$     W_{\alpha}= \alpha I + (1-\alpha)M $,
    where $\alpha \in (0,1)$ represents the ``holding probability,'' times 
    an input seed vector.
\end{description}
In each of these cases, there is a parameter ($t$, $\gamma$, and the number 
of steps of the Lazy Random Walk) that controls the ``aggressiveness'' of 
the dynamics and thus how quickly the diffusive process equilibrates; and 
there is an input ``seed'' distribution vector.
Thus, \emph{e.g.}, if one is interested in global spectral graph 
partitioning, then this seed vector could be a vector with entries drawn 
from $\{-1,+1\}$ uniformly at random, while if one is interested in local 
spectral graph 
partitioning~\cite{Spielman:2004,andersen06local,Chung07_heatkernelPNAS,MOV09_TRv2},
then this vector could be the indicator vector of a small ``seed set'' of 
nodes.
See Appendix~\ref{sxn:local-partitioning} for a brief discussion of local 
and global spectral partitioning in this context. 

Mahoney and Orecchia showed that these three dynamics arise as solutions to 
SDPs of the~form
\begin{equation}
\begin{aligned}
  & \underset{X}{\text{minimize}}
  & & \mathrm{Tr}(L X) + \tfrac{1}{\eta} G(X) \\
  & \text{subject to}
  & & X \succeq 0, \\
  & & & \mathrm{Tr}(X) = 1, \\
  & & & X D^{1/2} 1 = 0,
\end{aligned}
\label{eqn:mo-reg-sdp}
\end{equation}
where $G$ is a penalty function (shown to be the 
generalized entropy, the log-determinant, and a certain matrix-$p$-norm, 
respectively~\cite{MO11-implementing}) and where $\eta$ is a 
parameter related to the aggressiveness of the diffusive 
process~\cite{MO11-implementing}.  
Conversely, solutions to the regularized SDP of~(\ref{eqn:mo-reg-sdp}) for 
appropriate values of $\eta$ can be computed \emph{exactly} by running one 
of the above three diffusion-based procedures.
Notably, when $G = 0$, the solution to the SDP 
of~(\ref{eqn:mo-reg-sdp}) is $u u'$, where $u$ is the smallest nontrivial 
eigenvector of $L$.  
More generally and in this precise sense, the Heat Kernel, PageRank, and Lazy 
Random Walk dynamics can be seen as ``regularized'' versions of spectral 
clustering and Laplacian eigenvector computation.
Intuitively, the function $G(\cdot)$ is acting as a penalty function, in a 
manner analogous to the $\ell_2$ or $\ell_1$ penalty in Ridge regression or 
Lasso regression, and by running one of these three dynamics one is 
\emph{implicitly} making assumptions about the form of $G(\cdot)$.
In this paper, we provide a statistical framework to make that intuition 
precise.

\vspace{-2mm}
\section{A statistical framework for regularized graph estimation}
\label{snx:framework}
\vspace{-1mm}

Here, we will lay out a simple Bayesian framework for estimating a graph 
Laplacian.
Importantly, this framework will allow for regularization by incorporating 
prior information.

\vspace{-2mm}
\subsection{Analogy with regularized linear regression}
\label{S:regression}

It will be helpful to keep in mind the Bayesian interpretation of 
regularized linear regression.  
In that context, we observe $n$ predictor-response pairs in 
$\reals^p \times \reals$, denoted $(x_1, y_1), \dotsc, (x_n, y_n)$; 
the goal is to find a vector $\beta$ such that $\beta' x_i \approx y_i$.  
Typically, we choose $\beta$ by minimizing the residual sum of squares,
\emph{i.e.}, $F(\beta)=\RSS(\beta) = \sum_i \| y_i - \beta' x_i \|_2^2$, or 
a penalized version of it.  
For Ridge regression, we minimize $F(\beta) + \lambda \|\beta\|_2^2$; while 
for Lasso regression, we minimize $F(\beta) + \lambda \|\beta\|_1$.  

The additional terms in the optimization criteria (\emph{i.e.}, 
$\lambda \|\beta\|_2^2$ and $\lambda \|\beta\|_1$) are called penalty 
functions; and adding a penalty function to the optimization criterion can 
often be interpreted as incorporating prior information about $\beta$.  
For example, we can model $y_1, \dotsc, y_n$ as independent random 
observations with distributions dependent on $\beta$.  
Specifically, we can suppose $y_i$ is a Gaussian random variable with mean 
$\beta' x_i$ and known variance $\sigma^2$.  
This induces a conditional density for the vector $y = (y_1, \dotsc, y_n)$:
\begin{equation}\label{E:regression-density}
  p(y \mid \beta)
    \propto
     \exp\{ -\tfrac{1}{2 \sigma^2} F(\beta) \},
\end{equation}
where the constant of proportionality depends only on $y$ and $\sigma$.
Next, we can assume that $\beta$ itself is random, drawn from a 
distribution with density $p(\beta)$.  
This distribution is called a prior, since it encodes prior knowledge about 
$\beta$.  
Without loss of generality, the prior density can be assumed to take the form
\begin{equation}\label{E:regression-prior}
  p(\beta) \propto \exp\{ -U(\beta) \}.
\end{equation}
Since the two random variables are dependent, upon observing $y$, we have 
information about $\beta$.  
This information is encoded in the posterior density, $p(\beta \mid y)$, 
computed via Bayes' rule as
\begin{equation}\label{E:regression-posterior}
  p(\beta \mid y)
    \propto p(y \mid \beta) \, p(\beta)
    \propto \exp\{ -\tfrac{1}{2 \sigma^2} F(\beta) - U(\beta) \}.
\end{equation}
The MAP estimate of $\beta$ is the value that maximizes $p(\beta \mid y)$;  
equivalently, it is the value of $\beta$ that minimizes 
$-\log p(\beta \mid y)$.  
In this framework, we can recover the solution to Ridge regression or Lasso 
regression by setting
$U(\beta) = \tfrac{\lambda}{2 \sigma^2} \| \beta \|_2^2$ or
$U(\beta) = \tfrac{\lambda}{2 \sigma^2} \| \beta \|_1$, respectively.  
Thus, Ridge regression can be interpreted as imposing a Gaussian prior on 
$\beta$, and Lasso regression can be interpreted as imposing a 
double-exponential prior on $\beta$.

\vspace{-1mm}
\subsection{Bayesian inference for the population Laplacian}
\label{S:bayesian-laplacian}

For our problem, suppose that we have a connected graph with $n$ nodes; or, 
equivalently, that we have $L$, the normalized Laplacian of that graph.
We will view this observed graph Laplacian, $L$, as a ``sample'' Laplacian, 
\emph{i.e.}, as random object whose distribution depends on a true 
``population'' Laplacian, $\mathcal{L}$.  
As with the linear regression example, this induces a conditional density 
for $L$, to be denoted $p(L \mid \mathcal{L})$.  
Next, we can assume prior information about the population Laplacian in the 
form of a prior density, $p(\mathcal{L})$; and, 
given the observed Laplacian, we can estimate the population Laplacian by 
maximizing its posterior density, $p(\mathcal{L} \mid L)$.

Thus, to apply the Bayesian formalism, we need to specify the conditional 
density of $L$ given $\mathcal{L}$.  
In the context of linear regression, we assumed that the observations 
followed a Gaussian distribution.  
A graph Laplacian is not just a single observation---it is a positive 
semidefinite matrix with a very specific structure.  
Thus, we will take $L$ to be a random object with expectation~$\mathcal{L}$,
where $\mathcal{L}$ is another normalized graph Laplacian.  
Although, in general, $\mathcal{L}$ can be distinct from $L$, we will 
require that the nodes in the population and sample graphs have the same 
degrees.
That is, if $d = \big(d(1), \dotsc, d(n)\big)$ denotes the ``degree 
vector'' of the graph, and $D = \diag\big(d(1), \dotsc, d(n)\big)$, then we
can define
\begin{equation}
\mathcal{X} = \{ X : X \succeq 0, \, X D^{1/2} 1 = 0, \, \rank(X) = n - 1 \} ,
\label{def:chi-set}
\end{equation}
in which case the population Laplacian and the sample Laplacian will both 
be members of $\mathcal{X}$.  
To model $L$, we will choose a distribution for positive semi-definite 
matrices analogous to the Gaussian distribution: a scaled Wishart matrix 
with expectation $\mathcal{L}$.  
Note that, although it captures the trait that $L$ is positive 
semi-definite, this distribution does not accurately model every feature of 
$L$.
For example, a scaled Wishart matrix does not necessarily have ones along 
its diagonal.  
However, the mode of the density is at $\mathcal{L}$, a Laplacian; and for 
large values of the scale parameter, most of the mass will be on matrices 
close to $\mathcal{L}$.
Appendix~\ref{sxn:justification} provides a more detailed heuristic 
justification for the use of the Wishart distribution.

To be more precise, let $m \geq n - 1$ be a scale parameter, and suppose 
that $L$ is distributed over $\mathcal{X}$ as a 
$\tfrac{1}{m} \mathrm{Wishart}(\mathcal{L}, m)$ random variable.
Then, $\E[L \mid \mathcal{L}] = \mathcal{L}$, and $L$ has conditional density
\begin{equation}\label{E:density}
  p(L \mid \mathcal{L})
    \propto
      \frac{\exp\{ -\frac{m}{2} \Tr(L \mathcal{L}^+)\}}
           {|\mathcal{L}|^{m/2}},
\end{equation}
where $|\cdot|$ denotes pseudodeterminant (product of nonzero eigenvalues).  
The constant of proportionality depends only on $L$, $d$, $m$, and $n$; and 
we emphasize that the density is supported on $\mathcal{X}$.
Eqn.~\eqref{E:density} is analogous to Eqn.~\eqref{E:regression-density}
in the linear regression context, with $1/m$, the inverse of the sample 
size parameter, playing the role of the variance parameter $\sigma^2$.
Next, suppose we have know that $\mathcal{L}$ is a random object drawn from 
a prior density $p(\mathcal{L})$.
Without loss of generality,
\begin{equation}
\label{E:prior}
  p(\mathcal{L}) \propto \exp\{ -U(\mathcal{L}) \},
\end{equation}
for some function $U$, supported on a subset 
$\mathcal{\bar X} \subseteq \mathcal{X}$.
Eqn.~\eqref{E:prior} is analogous to Eqn.~\eqref{E:regression-prior}
from the linear regression example.
Upon observing $L$, the posterior distribution for $\mathcal{L}$ is
\begin{equation}\label{E:posterior}
  p(\mathcal{L} \mid L)
    \propto p(L \mid \mathcal{L}) \, p (\mathcal{L})
    \propto \exp\{ -\tfrac{m}{2} \Tr(L \mathcal{L}^+)
                   + \tfrac{m}{2} \log |\mathcal{L}^+|
                   - U(\mathcal{L}) \},
\end{equation}
with support determined by $\mathcal{\bar X}$.  
Eqn.~\eqref{E:posterior} is analogous to Eqn.~\eqref{E:regression-posterior}
from the linear regression example.
If we denote by $\mathcal{\hat L}$ the MAP estimate of $\mathcal{L}$, then 
it follows that $\mathcal{\hat L}^+$ is the solution to the program
\begin{equation}
\begin{aligned}
  & \underset{X}{\text{minimize}}
  & & \Tr(L X) + \tfrac{2}{m} U(X^+) - \log |X| \\
  & \text{subject to}
  & & X \in \mathcal{\bar X} \subseteq \mathcal{X} .
\end{aligned}
\end{equation}
Note the similarity with Mahoney-Orecchia regularized SDP 
of~(\ref{eqn:mo-reg-sdp}).
In particular, if
\(
  \mathcal{\bar X} = \{ X : \Tr(X) = 1 \} \cap \mathcal{X},
\)
then the two programs are identical except for the factor of
$\log |X|$ in the optimization criterion.

\vspace{-2mm}
\section{A prior related to the PageRank procedure}
\label{sxn:priors}
\vspace{-1mm}

Here, we will present a prior distribution for the population Laplacian 
that will allow us to leverage the estimation framework of 
Section~\ref{snx:framework}; and we will show that the MAP estimate 
of $\mathcal{L}$ for this prior is related to the PageRank procedure via 
the Mahoney-Orecchia regularized SDP.  
Appendix~\ref{sxn:other-priors} presents priors that lead to the Heat 
Kernel and Lazy Random Walk in an analogous way; in both of these cases, 
however, the priors are data-dependent in the strong sense that they 
explicitly depend on the number of data points.

\vspace{-2mm}
\subsection{Prior density}

The prior we will present will be based on neutrality and invariance 
conditions; and it will be supported on $\mathcal{X}$, \emph{i.e.}, on the 
subset of positive-semidefinite matrices that was the support set for the 
conditional density defined in Eqn.~(\ref{E:density}).
In particular, recall that, in addition to being positive semi-definite, 
every matrix in the support set has rank $n - 1$ and satisfies 
$X D^{1/2} 1 = 0$.  
Note that because the prior depends on the data (via the orthogonality 
constraint induced by $D$), this is not a prior in the fully Bayesian 
sense; instead, the prior can be considered as part of an empirical or 
pseudo-Bayes estimation procedure.

The prior we will specify depends only on the eigenvalues of the normalized
Laplacian, or equivalently on the eigenvalues of the pseudoinverse of the 
Laplacian.  
Let $\mathcal{L}^+ = \tau \, O \Lambda O'$ be the spectral decomposition of 
the pseudoinverse of the normalized Laplacian $\mathcal{L}$, where $\tau \geq 0$
is a scale factor, $O \in \reals^{n \times n -1}$ is an orthogonal matrix, 
and $\Lambda = \diag\big(\lambda(1), \dotsc, \lambda({n-1})\big)$, where
$\sum_v \lambda(v) = 1$.
Note that the values $\lambda(1), \dotsc, \lambda(n-1)$ are unordered and 
that the vector $\lambda = \big(\lambda(1), \dotsc, \lambda(n-1)\big)$ lies
in the unit simplex.  
If we require that the distribution for $\lambda$ be exchangeable 
(invariant under permutations) and neutral ($\lambda(v)$ independent of the 
vector $\big(\lambda(u) / (1 - \lambda(v)) : u \neq v\big)$, for all $v$),
then the only non-degenerate possibility is that $\lambda$ is
Dirichlet-distributed with parameter vector 
$(\alpha, \ldots, \alpha)$~\cite{fabius1973two}.  
The parameter $\alpha$, to which we refer as the ``shape'' parameter, must 
satisfy $\alpha > 0$ for the density to be defined.  
In this case,
\begin{equation}\label{E:dirichlet-prior}
  p(\mathcal{L})
   \propto p(\tau) \prod_{v=1}^{n-1} \lambda(v)^{\alpha - 1},
\end{equation}
where $p(\tau)$ is a prior for $\tau$.
Thus, the prior weight on $\mathcal{L}$ only depends on $\tau$ and $\Lambda$.  
One implication is that the prior is ``nearly'' rotationally invariant, 
in the sense that $p(P' \mathcal{L} P) = p(\mathcal{L})$ for any 
rank-$(n-1)$ projection matrix $P$ satisfying $P D^{1/2} 1 = 0$.

\vspace{-2mm}
\subsection{Posterior estimation and connection to PageRank}
\label{S:posterior-density}

To analyze the MAP estimate associated with the prior of 
Eqn.~(\ref{E:dirichlet-prior}) and to explain its connection with the 
PageRank dynamics, the following proposition is crucial.

\begin{proposition}\label{P:map-sdp}
Suppose the conditional likelihood for $L$ given $\mathcal{L}$ is as 
defined in \eqref{E:density} and the prior density for $\mathcal{L}$ is as 
defined in \eqref{E:dirichlet-prior}.  
Define $\mathcal{\hat L}$ to be the MAP estimate of $\mathcal{L}$.
Then, $[\Tr(\mathcal{\hat L}^+)]^{-1} \mathcal{\hat L}^+$ solves the
Mahoney-Orecchia regularized SDP \eqref{eqn:mo-reg-sdp}, 
with $G(X) = -\log |X|$ and $\eta$ as given in Eqn.~(\ref{E:eta}) below.
\end{proposition}

\begin{proof}
For $\mathcal{L}$ in the support set of the posterior,
define $\tau = \Tr(\mathcal{L}^+)$ and $\Theta  = \tau^{-1}
\mathcal{L}^+$, so that $\Tr(\Theta) = 1$.  Further, $\rank(\Theta) = n - 1$.
Express the prior in the form of Eqn.~\eqref{E:prior} with function
$U$ given by
\[
  U(\mathcal{L})
    = -\log \{ p(\tau) \, |\Theta|^{\alpha - 1} \}
    = -(\alpha - 1) \log |\Theta| - \log p(\tau),
\]
where, as before, $|\cdot|$ denotes pseudodeterminant.  Using
\eqref{E:posterior} and the relation
$|\mathcal{L}^+| = \tau^{n-1} |\Theta|$, the posterior
density for $\mathcal{L}$ given $L$ is
\[
  p(\mathcal{L} \mid L)
    \propto
      \exp\Big\{
        -\tfrac{m \tau}{2} \Tr(L \Theta)
        +\tfrac{m + 2(\alpha - 1)}{2} \log | \Theta |
        + g(\tau)\Big\},
\]
where
\(
  g(\tau)
    = \tfrac{m (n-1)}{2} \log \tau
        + \log p(\tau).
\)
Suppose $\mathcal{\hat L}$ maximizes the
posterior likelihood.  Define $\hat \tau = \Tr(\mathcal{\hat L}^+)$
and $\hat \Theta = [\hat \tau]^{-1} \mathcal{\hat L}^{+}$.
In this case, $\hat \Theta$ must minimize the quantity
\(
  \Tr(L \hat \Theta) - \tfrac{1}{\eta} \log |\hat \Theta|,
\)
where
\begin{equation}\label{E:eta}
  \eta
    = \frac{m \hat \tau}{m + 2(\alpha - 1)}.
\end{equation}
Thus $\hat \Theta$ solves the regularized SDP \eqref{eqn:mo-reg-sdp} with 
$G(X) = - \log |X|$.
\end{proof}

Mahoney and Orecchia showed that the solution to
\eqref{eqn:mo-reg-sdp} with $G(X) = - \log |X|$ is closely related to
the PageRank matrix, $R_\gamma$, defined in Eqn.~\eqref{eqn:page-rank}.
By combining Proposition~\ref{P:map-sdp} with their result, we get that the 
MAP estimate of $\mathcal{L}$ satisfies
\( \mathcal{\hat L}^+ \propto D^{-1/2} R_\gamma D^{1/2}; \)
conversely,
\(
  R_\gamma
    \propto
      D^{1/2} \mathcal{\hat L}^+ D^{-1/2}.
\)
Thus, the PageRank operator of Eqn.~(\ref{eqn:page-rank}) can be viewed as a 
degree-scaled regularized estimate of the pseudoinverse of the Laplacian.  
Moreover, prior assumptions about the spectrum of the graph Laplacian have 
direct implications on the optimal teleportation parameter.  
Specifically Mahoney and Orecchia's Lemma~2 shows how $\eta$ is related to 
the teleportation parameter $\gamma$, and Eqn. \eqref{E:eta} shows how the
optimal $\eta$ is related to prior assumptions about the Laplacian.

\vspace{-2mm}
\section{Empirical evaluation}
\label{sxn:empirical}
\vspace{-1mm}

In this section, we provide an empirical evaluation of the performance of 
the regularized Laplacian estimator, compared with the unregularized 
estimator.  
To do this, we need a ground truth population Laplacian $\mathcal{L}$ and a 
noisily-observed sample Laplacian $L$. 
Thus, in Section~\ref{S:prior-evaluation}, we construct a family of 
distributions for $\mathcal{L}$; importantly, this family will be able to 
represent both low-dimensional graphs and expander-like graphs.  
Interestingly, the prior of Eqn.~\eqref{E:dirichlet-prior} captures some of 
the qualitative features of both of these types of graphs (as the shape 
parameter is varied).  
Then, in Section~\ref{S:estimation}, we describe a sampling procedure for 
$L$ which, superficially, has no relation to the scaled Wishart conditional 
density of Eqn.~\eqref{E:density}.  
Despite this model misspecification, the regularized estimator 
$\hat L_{\eta}$ outperforms $L$ for many choices of the regularization 
parameter $\eta$.

\vspace{-2mm}
\subsection{Ground truth generation and prior evaluation}
\label{S:prior-evaluation}

The ground truth graphs we generate are motivated by the Watts-Strogatz 
``small-world'' model~\cite{watts98collective}.
To generate a ground truth population Laplacian, 
$\mathcal{L}$---equivalently, a population graph---we start with a 
two-dimensional lattice of width $w$ and height $h$, and thus $n = w h$ nodes.  
Points in the lattice are connected to their four nearest neighbors, making 
adjustments as necessary at the boundary.  
We then perform $s$ edge-swaps: for each swap, we choose two edges 
uniformly at random and then we swap the endpoints.  
For example, if we sample edges $i_1 \sim j_1$ and $i_2 \sim j_2$, then we
replace these edges with $i_1 \sim j_2$ and $i_2 \sim j_1$.  
Thus, when $s = 0$, the graph is the original discretization of a 
low-dimensional space; and as $s$ increases to infinity, the graph becomes 
more and more like a uniformly chosen $4$-regular graph (which is an 
expander~\cite{HLW06_expanders} and which bears similarities with an 
Erd\H{o}s-R\'{e}nyi random graph~\cite{Bollobas85}).
Indeed, each edge swap is a step of the Metropolis algorithm toward a 
uniformly chosen random graph with a fixed degree sequence.  
For the empirical evaluation presented here, $h = 7$ and $w = 6$; but the 
results are qualitatively similar for other values.

\begin{figure}[t]
    \centering
    \subfigure[Dirichlet distribution order statistics.]{
      \makebox{\label{F:dirichlet}\includegraphics[scale=0.6]{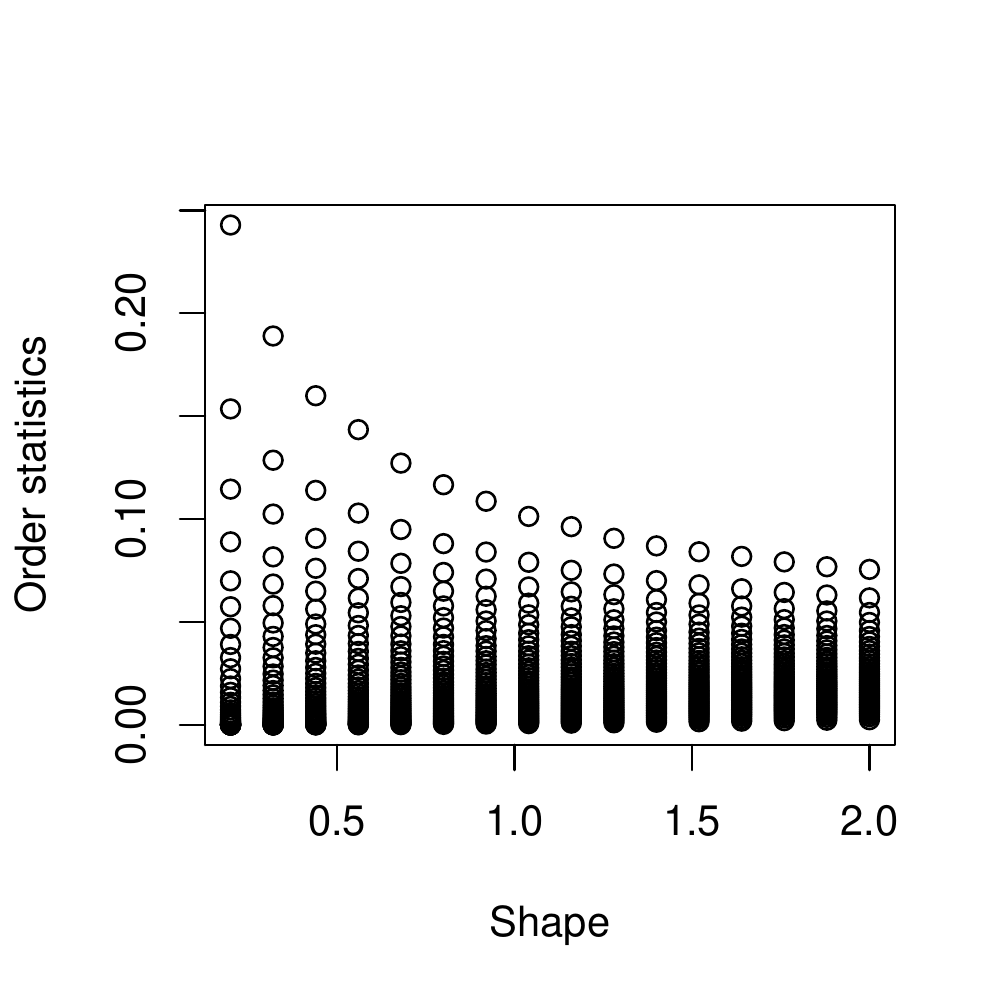}}
    }
    \subfigure[Spectrum of the inverse Laplacian.]{
      \makebox{\label{F:interpolate}\includegraphics[scale=0.6]{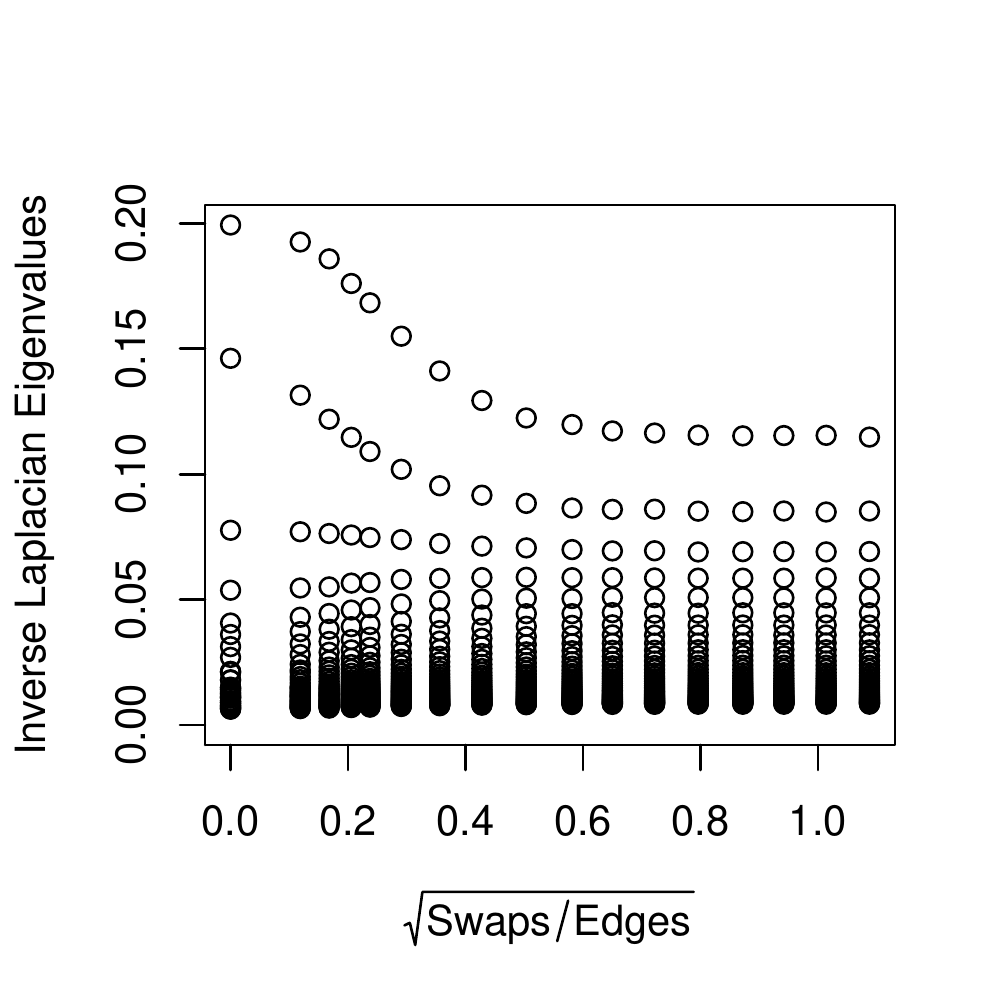}}
    }
\vspace{-0.5em}
\caption{Analytical and empirical priors.
\ref{F:dirichlet} shows the Dirichlet distribution order statistics versus the shape parameter; and
\ref{F:interpolate} shows the spectrum of $\Theta$ as a function of the 
rewiring parameter.  }
\vspace{-1em}
\label{fig:priors}
\end{figure}

Figure~\ref{fig:priors} compares the expected order statistics (sorted 
values) for the Dirichlet prior of Eqn.~\eqref{E:dirichlet-prior} with the 
expected eigenvalues of $\Theta  = \mathcal{L}^+ / \Tr(\mathcal{L}^+) $ for 
the small-world model.
In particular, in Figure~\ref{F:dirichlet}, we show the behavior of the 
order statistics of a Dirichlet distribution on the $(n-1)$-dimensional 
simplex with scalar shape parameter $\alpha$, as a function of $\alpha$.  
For each value of the shape $\alpha$, we generated a random 
$(n-1)$-dimensional Dirichlet vector, $\lambda$, with parameter vector 
$(\alpha, \dotsc, \alpha)$; we computed the $n-1$ order statistics of 
$\lambda$ by sorting its components; and we repeated this procedure for 
$500$ replicates and averaged the values.  
Figure~\ref{F:interpolate} shows a corresponding plot for the ordered 
eigenvalues of $\Theta$.
For each value of $s$ (normalized, here, by the number of edges $\mu$, 
where $\mu = 2 w h - w - h=71$), we generated the normalized Laplacian, 
$\mathcal{L}$, corresponding to the random $s$-edge-swapped grid; we 
computed the $n-1$ nonzero eigenvalues of $\Theta$; and we performed $1000$ 
replicates of this procedure and averaged the resulting eigenvalues. 

Interestingly, the behavior of the spectrum of the small-world model as
the edge-swaps increase is qualitatively quite similar to the behavior of 
the Dirichlet prior order statistics as the shape parameter $\alpha$ 
increases.
In particular, note that for small values of the shape parameter $\alpha$ 
the first few order-statistics are well-separated from the rest; and that 
as $\alpha$ increases, the order statistics become concentrated around 
$1/(n-1)$.
Similarly, when the edge-swap parameter $s = 0$, the top two eigenvalues 
(corresponding to the width-wise and height-wise coordinates on the grid) 
are well-separated from the bulk; as $s$ increases, the top eigenvalues 
quickly merge into the bulk; and eventually, as $s$ goes to infinity, the 
distribution becomes very close that that of a uniformly chosen $4$-regular 
graph.

\vspace{-2mm}
\subsection{Sampling procedure, estimation performance, and optimal regularization behavior}
\label{S:estimation}

\begin{figure}[h]
  \centering
  \subfigure[$m/\mu = 0.2$ and $s = 4$.]{
    \makebox{\label{fig:perf-frob:a}
             \includegraphics[scale=0.34]{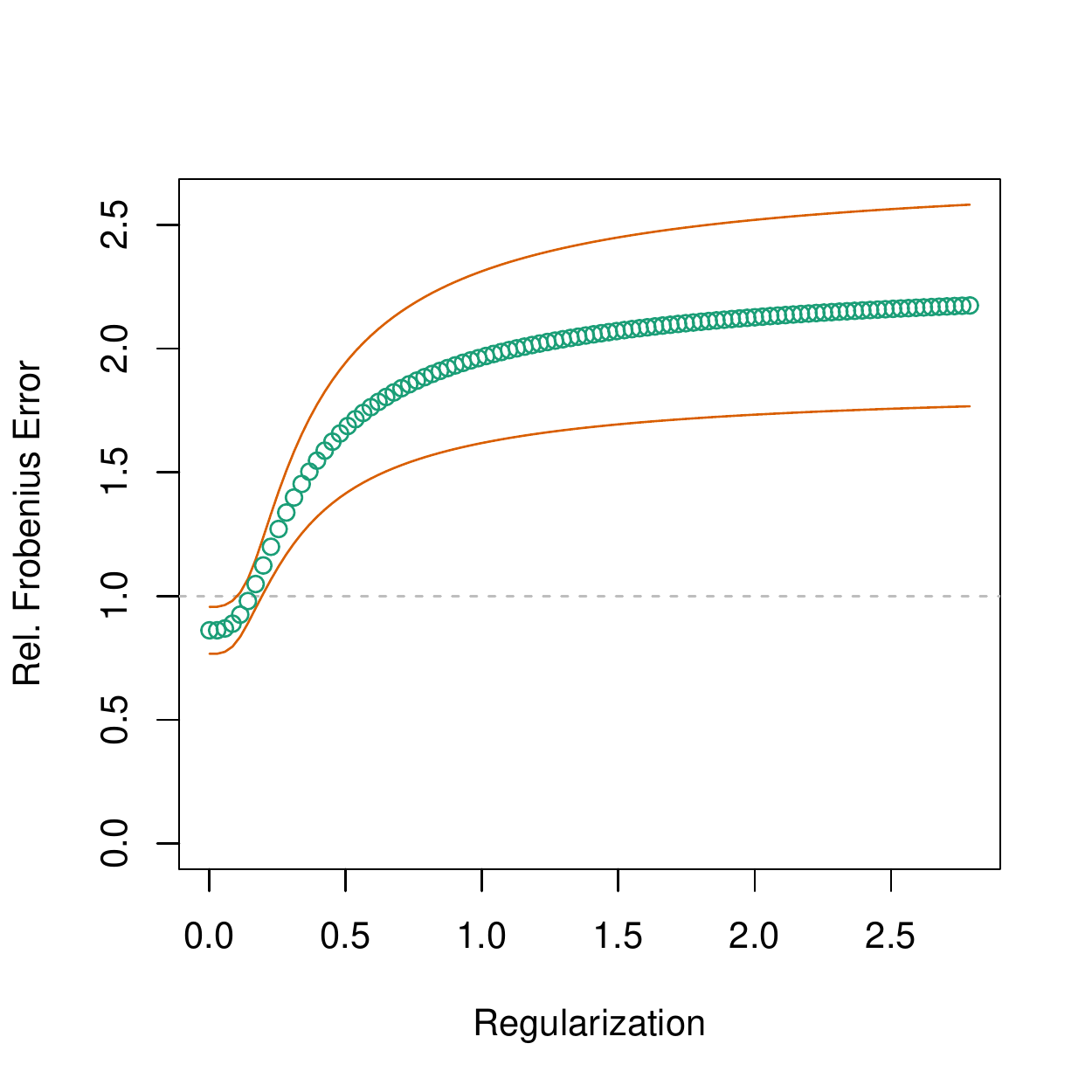}}
  }
  \subfigure[$m/\mu = 1.0$ and $s = 4$.]{
    \makebox{\label{fig:perf-frob:b}
             \includegraphics[scale=0.34]{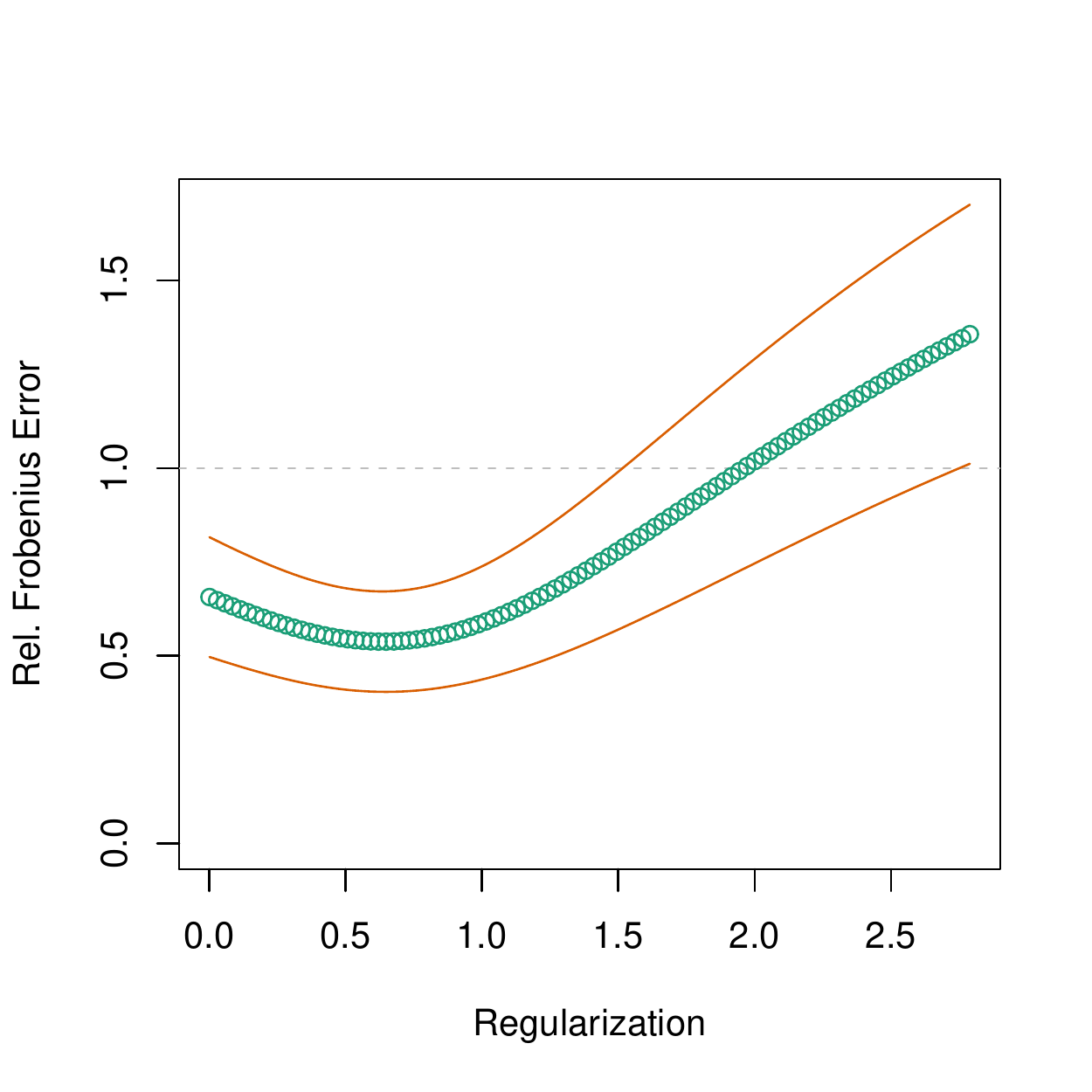}}
  }
  \subfigure[$m/\mu = 2.0$ and $s = 4$.]{
    \makebox{\label{fig:perf-frob:c}
             \includegraphics[scale=0.34]{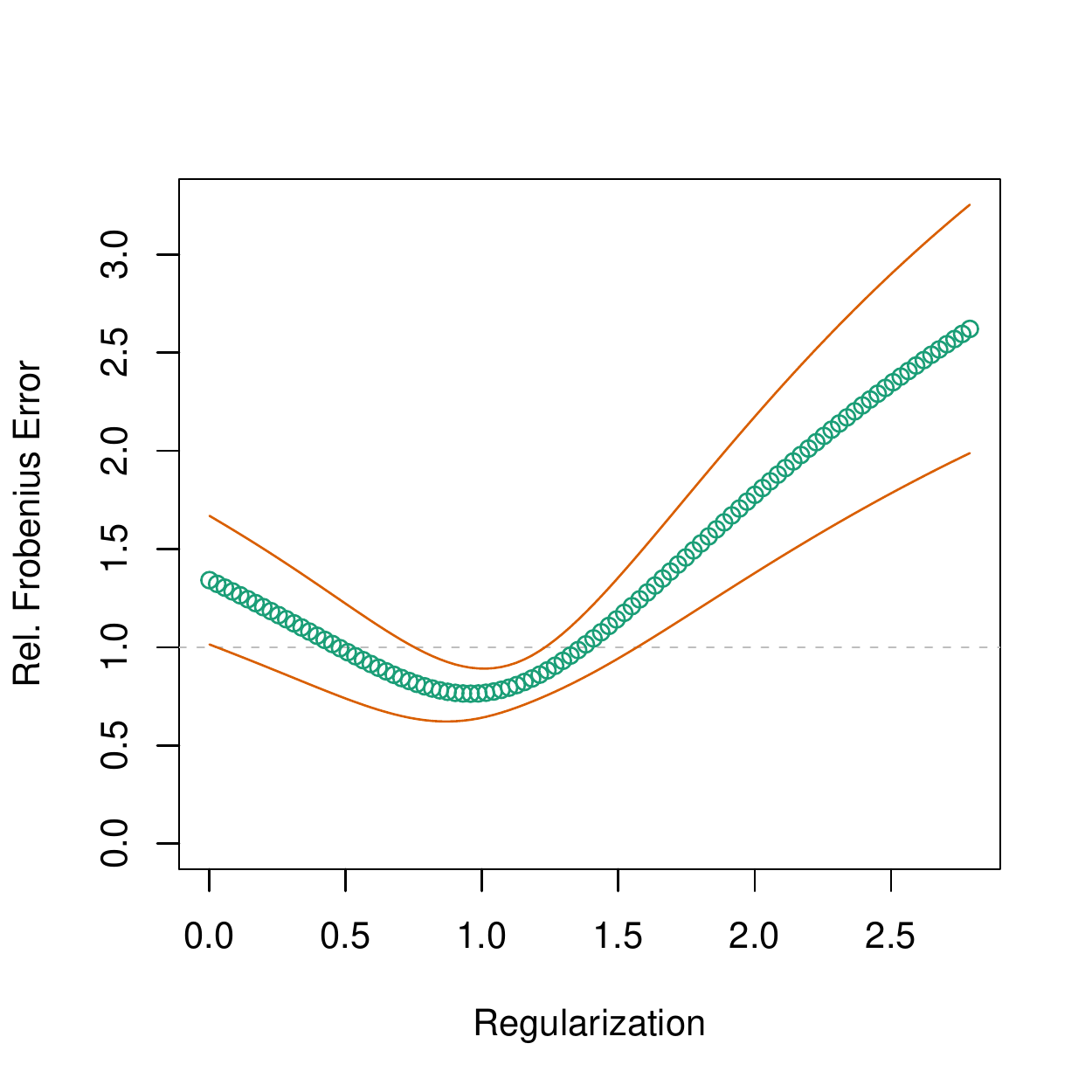}}
  } \\
\vspace{-1em}
  \subfigure[$m/\mu = 2.0$ and $s = 0$.]{
    \makebox{\label{fig:perf-frob:d}
             \includegraphics[scale=0.34]{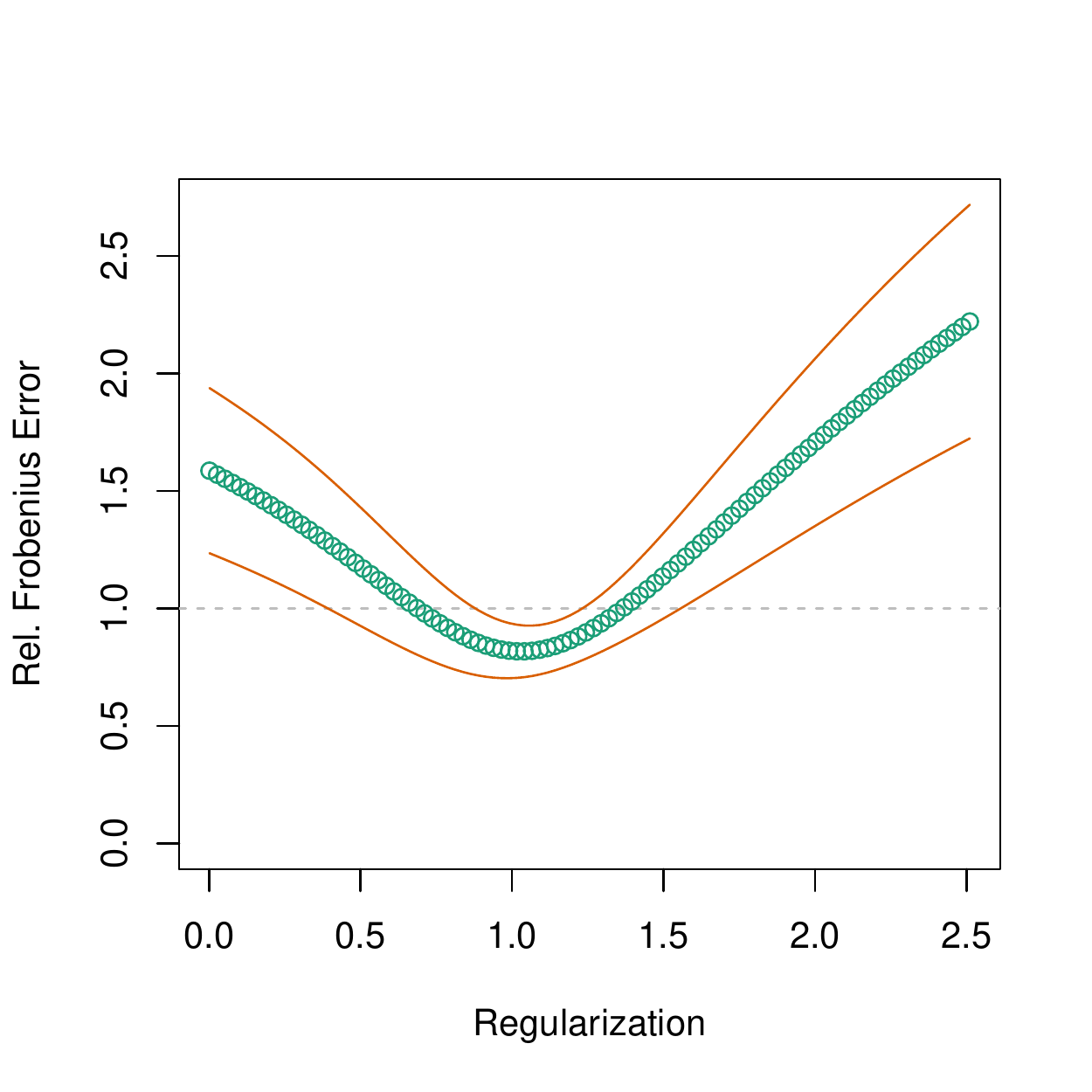}}
  }
  \subfigure[$m/\mu = 2.0$ and $s = 32$.]{
    \makebox{\label{fig:perf-frob:e}
             \includegraphics[scale=0.34]{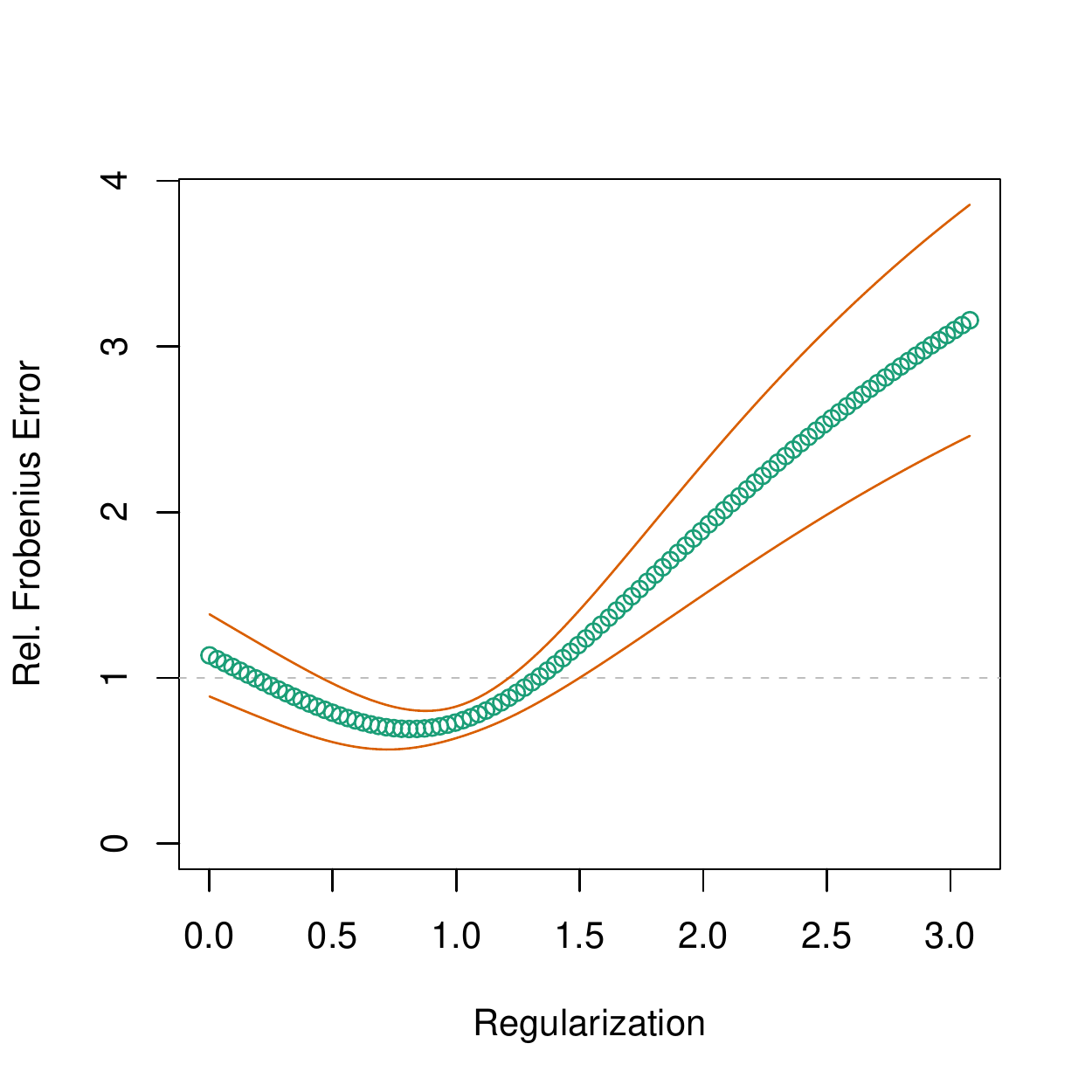}}
  }
  \subfigure[Optimal $\eta^\ast / \bar \tau$.]{
    \makebox{\label{F:optimal-frob}
             \includegraphics[scale=0.34]{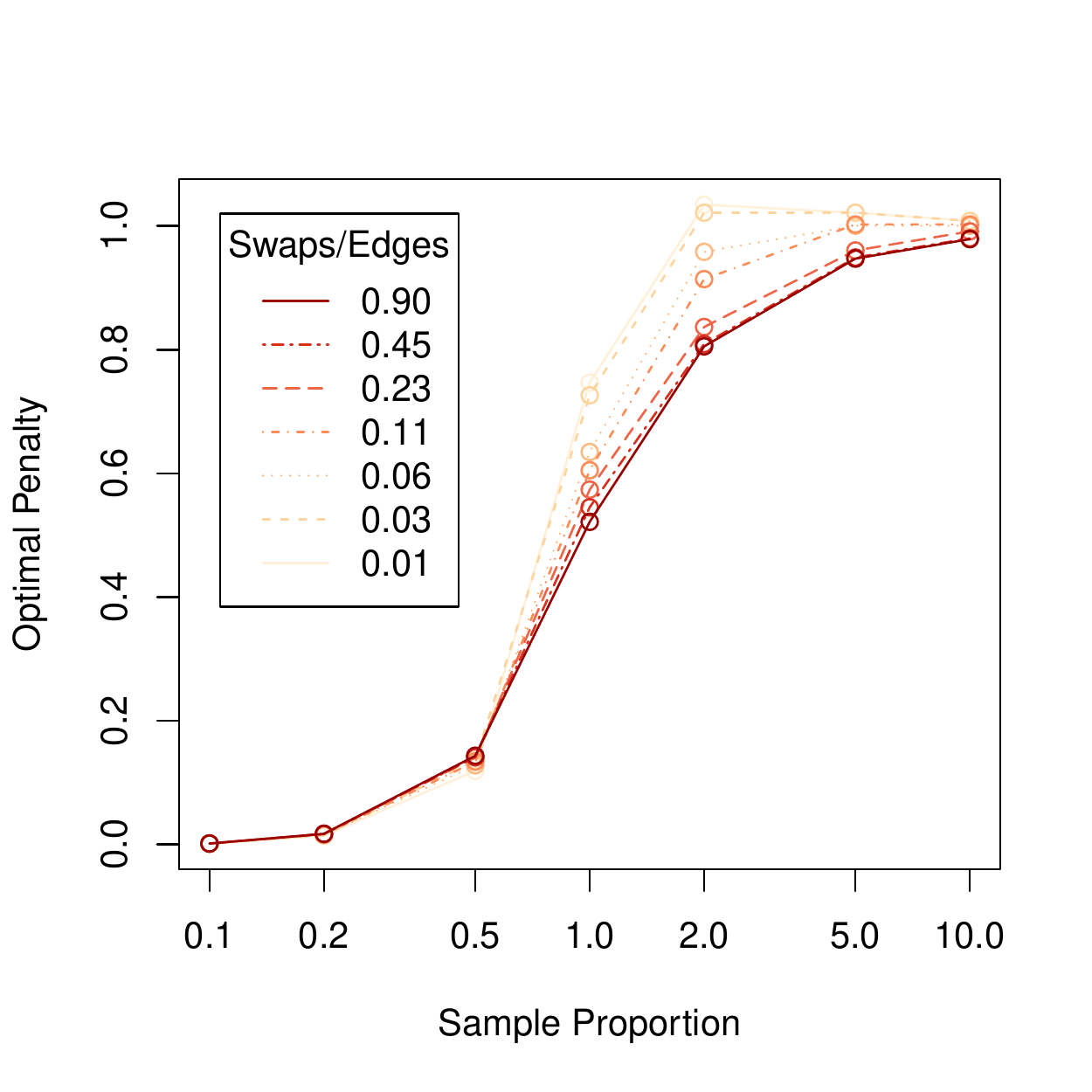}}
  }
\vspace{-1em}
\caption{Regularization performance.
\ref{fig:perf-frob:a} through \ref{fig:perf-frob:e} plot the relative 
Frobenius norm error, versus the (normalized) regularization parameter 
$\eta / \bar \tau$.
Shown are plots for various values of the (normalized) number of edges, 
$m/\mu$, and the edge-swap parameter, $s$.
Recall that the regularization parameter in the regularized 
SDP~(\ref{eqn:mo-reg-sdp}) is $1/\eta$, and thus smaller values along the 
X-axis correspond to stronger regularization.
\ref{F:optimal-frob} plots the optimal regularization parameter 
$\eta^\ast / \bar \tau$ as a function of sample proportion for different
fractions of edge swaps.  }
\label{fig:perf-frob}
\end{figure}

Finally, we evaluate the estimation performance of a regularized estimator 
of the graph Laplacian and compare it with an unregularized estimate.  
To do so, we construct the population graph $\mathcal{G}$ and its Laplacian 
$\mathcal{L}$, for a given value of $s$, as described in 
Section~\ref{S:prior-evaluation}.  
Let $\mu$ be the number of edges in $\mathcal{G}$. 
The sampling procedure used to generate the observed graph $G$ and its 
Laplacian $L$ is parameterized by the sample size $m$.
(Note that this parameter is analogous to the Wishart scale parameter in
Eqn.~\eqref{E:density}, but here we are sampling from a different 
distribution.)
We randomly choose $m$ edges with replacement from $\mathcal{G}$; and we 
define sample graph $G$ and corresponding Laplacian $L$ by setting the weight 
of $i \sim j$ equal to the number of times we sampled that edge.  
Note that the sample graph $G$ over-counts some edges in $\mathcal{G}$ and 
misses others.

We then compute the regularized estimate $\mathcal{\hat L}_\eta$, up to a 
constant of proportionality, by solving (implicitly!) the Mahoney-Orecchia
regularized SDP~\eqref{eqn:mo-reg-sdp} with $G(X) = -\log|X|$.  
We define the unregularized estimate $\hat L$ to be equal to the observed 
Laplacian, $L$.  
Given a population Laplacian $\mathcal{L}$, we define
$\tau = \tau(\mathcal{L}) = \Tr(\mathcal{L}^+)$ and
$\Theta = \Theta(\mathcal{L}) = \tau^{-1} \mathcal{L}^+$. 
We define $\hat \tau_\eta$, $\hat \tau$, $\hat \Theta_\eta$, and 
$\hat \Theta$ similarly to the population quantities.  
Our performance criterion is the relative Frobenius error
$\|\Theta - \hat \Theta_\eta\|_\mathrm{F} / \|\Theta - \hat \Theta
\|_\mathrm{F}$, where $\|\cdot\|_\mathrm{F}$ denotes the Frobenius norm
($\|A\|_\mathrm{F} = [\Tr(A' A)]^{1/2}$).  
Appendix~\ref{S:spectral-performance} presents similar results when the
performance criterion is the relative spectral norm error.

Figures~\ref{fig:perf-frob:a}, \ref{fig:perf-frob:b}, 
and~\ref{fig:perf-frob:c} show the regularization performance when $s=4$
(an intermediate value) for three different values of $m/\mu$.
In each case, the mean error and one standard deviation around it are 
plotted as a function of $\eta/\bar \tau$, as computed from 100 replicates;  
here, $\bar \tau$ is the mean value of $\tau$ over all replicates.
The implicit regularization clearly improves the performance of the 
estimator for a large range of $\eta$ values.  
(Note that the regularization parameter in the regularized 
SDP~(\ref{eqn:mo-reg-sdp}) is $1/\eta$, and thus smaller values along the
X-axis correspond to stronger regularization.)
In particular, when the data are very noisy, \emph{e.g.}, when $m/\mu=0.2$, 
as in Figure~\ref{fig:perf-frob:a}, improved results are seen only for very 
strong regularization;
for intermediate levels of noise, \emph{e.g.},  $m/\mu=1.0$, as in 
Figure~\ref{fig:perf-frob:b}, (in which case $m$ is chosen such that $G$ 
and $\mathcal{G}$ have the same number of edges counting multiplicity), 
improved performance is seen for a wide range of values of $\eta$; and for 
low levels of noise, Figure~\ref{fig:perf-frob:c} illustrates that improved 
results are obtained for moderate levels of implicit regularization.
Figures~\ref{fig:perf-frob:d} and~\ref{fig:perf-frob:e} illustrate similar
results for $s=0$ 
and 
$s=32$. 

As when regularization is implemented explicitly, in all these cases, we 
observe a ``sweet spot'' where there is an optimal value for the implicit 
regularization parameter.
Figure~\ref{F:optimal-frob} illustrates how the optimal choice of $\eta$ 
depends on parameters defining the population Laplacians and sample 
Laplacians.
In particular, it illustrates how $\eta^\ast$, the optimal value of 
$\eta$ (normalized by $\bar \tau$), depends on the sampling proportion 
$m / \mu$ and the swaps per edges $s / \mu$.
Observe that as the sample size $m$ increases, $\eta^\ast$ converges
monotonically to $\bar \tau$; and, further, that higher values of $s$ 
(corresponding to more expander-like graphs) correspond to higher values 
of $\eta^\ast$.
Both of these observations are in direct agreement with Eqn.~\eqref{E:eta}.

\vspace{-2mm}
\section{Conclusion}
\label{sxn:conc}
\vspace{-1mm}

We have provided a statistical interpretation for the observation that 
popular diffusion-based procedures to compute a quick approximation to the 
first nontrivial eigenvector of a data graph Laplacian exactly solve a 
certain regularized version of the problem.
One might be tempted to view our results as ``unfortunate,'' in that it is 
not straightforward to interpret the priors presented in this paper.
Instead, our results should be viewed as making explicit the implicit prior 
assumptions associated with making certain decisions (that are \emph{already} 
made in practice) to speed up computations.

Several extensions suggest themselves.
The most obvious might be to try to obtain Proposition~\ref{P:map-sdp} with 
a more natural or empirically-plausible model than the Wishart distribution; 
to extend the empirical evaluation to much larger and more realistic data 
sets; to apply our methodology to other widely-used approximation 
procedures; and to characterize when implicitly regularizing an eigenvector 
leads to better statistical behavior in downstream applications where that 
eigenvector is used.
More generally, though, we expect that understanding the 
algorithmic-statistical tradeoffs that we have illustrated will become 
increasingly important in very large-scale data analysis applications.

\newpage


\newpage
\appendix

\section{Relationship with local and global spectral graph partitioning}
\label{sxn:local-partitioning}
\vspace{-2mm}

In this section, we briefly describe the connections between our results 
and global and local versions of spectral partitioning, which were a 
starting point for this work.

The idea of spectral clustering is to approximate the best partition of the 
vertices of a connected graph into two pieces by using the nontrivial 
eigenvectors of a Laplacian (either the combinatorial or the normalized
Laplacian).  
This idea has had a long history~\cite{Donath:1973,fiedler75B,spielman96_spectral,guatterymiller98,ShiMalik00_NCut}.
The simplest version of spectral clustering involves computing the first 
nontrivial eigenvector (or another vector with Rayleigh quotient close to 
that of the first nontrivial eigenvector) of $L$, and ``sweeping'' over 
that vector. 
For example, 
one can take that vector, call it $x$, and, for a given threshold $K$, 
define the two sets of the partition as
\vspace{-2mm}
\begin{align*}
  C(x,K)      &= \{ i : x_i \ge K \}, \quad \text{and} \\
  \bar C(x,K) &= \{ i : x_i < K \}.
\end{align*}
Typically, this vector $x$ is computed by calling a black-box solver, but it 
could also be approximated with an iteration-based method (such as the Power
Method or Lanczos Method) or a random walk-based method (such as running a 
diffusive procedure or PageRank-based procedure to the asymptotic state).
Far from being a ``heuristic,'' 
this procedure provides a \emph{global 
partition} that (via Cheeger's inequality and for an appropriate choice of 
$K$) satisfies provable quality-of-approximation guarantees with respect to 
the combinatorial optimization problem of finding the best ``conductance'' 
partition in the entire 
graph~\cite{Mihail,spielman96_spectral,guatterymiller98}.

Although spectral clustering reduces a graph to a single vector---the 
smallest nontrivial eigenvector of the graph's Laplacian---and then 
clusters the nodes using the information in that vector, 
it is possible to obtain much more information about the graph by looking 
at more than one eigenvector of the Laplacian.  
In particular, the elements of the pseudoinverse of the combinatorial 
Laplacian, $L_0^{+}$, give local (\emph{i.e.}, node-specific) information 
about random walks on the graph.
The reason is that the pseudoinverse $L_0^+$  of the Laplacian is closely 
related to random walks on the graph.  
See, \emph{e.g.}~\cite{chebotarev1998proximity} for details.  
For example, 
it is known that the
quantity $L_0^{+} (u,u) + L_0^{+} (v,v) - L_0^{+} (u,v) - L_0^{+} (v,u)$ is
proportional to the commute time, a symmetrized version of the length of 
time before a random walker started at node $u$ reaches node $v$, whenever 
$u$ and $v$ are in the same connected component~\cite{chandra1989electrical}.  
Similarly, the elements of the pseudoinverse of the \emph{normalized} 
Laplacian give degree-scaled measures of proximity between the nodes of a 
graph.
It is likely that $L^{+}(u,v)$ has a probabilistic interpretation in terms 
of random walks on the graph, along the lines of our methodology, but we 
are not aware of any such interpretation.
From this perspective, given $L^{+}$ and a cutoff value, $K$, we can 
define a \emph{local partition} around node $u$ via 
$P_K(u) = \{ v : L^{+}(u,v) > K \}$.  
(Note that if $v$ is in $P_K(u)$, then $u$ is in $P_K(v)$; in addition, if 
the graph is disconnected, then there exists a $K$ such that $u$ and $v$ 
are in the same connected component iff $v \in P_K(u)$.)  
We call clustering procedures based on this idea \emph{local spectral 
partitioning}.

Although the na\"{i}ve way of performing this local spectral partitioning, 
\emph{i.e.}, to compute $L^{+}$ explicitly, is prohibitive for anything but 
very small graphs, these ideas form the basis for very fast local spectral 
clustering methods that employ truncated diffusion-based procedures to 
compute localized vectors with which to partition.
For example, this idea can be implemented by performing a diffusion-based 
procedure with an input seed distribution vector localized on a node $u$ and 
then sweeping over the resulting vector.
This idea was originally introduced in~\cite{Spielman:2004} as a 
diffusion-based operational procedure that was local in a very strong sense 
and that led to Cheeger-like bounds analogous to those obtained with the 
usual global spectral partitioning; and this was extended and improved 
by~\cite{andersen06local,Chung07_heatkernelPNAS}.
In addition, an optimization perspective on this was provided 
by~\cite{MOV09_TRv2}.
Although~\cite{MOV09_TRv2} is local in a weaker sense, it does obtain local 
Cheeger-like guarantees from an explicit locally-biased optimization 
problem, and it provides an optimization ansatz that may be interpreted as 
a ``local eigenvector.''
See~\cite{Spielman:2004,andersen06local,Chung07_heatkernelPNAS,MOV09_TRv2} 
for details.
Understanding the relationship between the ``operational procedure versus
optimization ansatz'' perspectives was the origin 
of~\cite{MO11-implementing} and thus of this work.

\vspace{-2mm}
\section{Heuristic justification for the Wishart density}
\label{sxn:justification}
\vspace{-2mm}

In this section, we describe a sampling procedure for $L$ which, in a very 
crude sense, leads approximately to a conditional Wishart density for 
$p(L \mid \mathcal{L})$.

Let $G$ be a graph with vertex set $V = \{ 1, 2, \dotsc, n \}$, edge set 
$E = V \times V$ equipped with the equivalence relation $(u,v) = (v,u)$.  
Let $\omega$ be an edge weight function, and let $\mathcal{L}_0$ and 
$\mathcal{L}$ be the corresponding combinatorial and normalized Laplacians.  
Let $\Delta$ be a diagonal matrix with
$\Delta(u,u) = \sum_{v} \omega(u,v)$, so that
$\mathcal{L} = \Delta^{-1/2} \mathcal{L}_0 \Delta^{-1/2}$.  Suppose
the weights are scaled such that $\sum_{(u,v) \in E} \omega(u,v) = 1$,
and suppose further that $\Delta(u,u) > 0$.
We refer to $\omega(u,v)$ as the population weight of edge $(u,v)$.

A simple model for the sample graph is as follows: we sample $m$ edges
from $E$, randomly chosen according to the population weight function.
That is, we see edges $(u_1, v_1), (u_2, v_2), \dotsc, (u_m, v_m)$,
where the edges are all drawn independently and identically such that
the probability of seeing edge $(u,v)$ is determined by $\omega$:
\[
  \prob_\omega\{ (u_1, v_1) = (u,v) \} = \omega(u,v).
\]
Note that we will likely see duplicate edges and not every edge with a
positive weight will get sampled.
Then, we construct a weight function from the sampled edges, called the
sample weight function, $w$, defined such that
\vspace{-2mm}
\[
  w(u,v) = \frac{1}{m} \sum_{i=1}^{m} 1\{ (u_i, v_i) = (u,v) \} ,
\vspace{-2mm}
\]
where $1\{\cdot\}$ is an indicator vector.
In turn, we construct a sample combinatorial Laplacian, $L_0$, defined such 
that
\[
  L_0(u,v)
    =
    \begin{cases}
      \sum_{w} w(u,w) &\text{when $u = v$,} \\
      -w(u,v) &\text{otherwise.}
    \end{cases}
\]
Let $D$ be a diagonal matrix such that
$D(u,u) = \sum_{v} w(u,v)$, and define $L = D^{-1/2} L_0 D^{-1/2}$.
Letting $\E_\omega$ denote expectation with respect to the probability
law $\prob_\omega$, note that $\E_\omega[w(u,v)] = \omega(u,v)$,
that $\E_\omega L_0 = \mathcal{L}_0$, and that $\E_\omega D = \Delta$.
Moreover, the strong law of large numbers guarantees that as $m$ increases,
these three quantities converge almost surely to their expectations.
Further, Slutzky's theorem guarantees that $\sqrt{m} (L - \mathcal{L})$ and
$\sqrt{m} \Delta^{-1/2} (L_0 - \mathcal{L}_0) \Delta^{-1/2}$ converge in
distribution to the same limit.  
We use this large-sample behavior to
approximate the the distribution of $L$ by the distribution of
$\Delta^{-1/2} L_0 \Delta^{-1/2}$.  Put simply, we treat the degrees as known.

The distribution of $L_0$ is completely determined by the edge
sampling scheme laid out above.  However, the exact form for the
density involves an intractable combinatorial sum.  
Thus, we appeal to a
crude approximation for the conditional density.
The approximation works as follows:
\begin{enumerate}
\item For $i = 1, \dotsc, m$, define $x_i \in \reals^n$ such that
  \[
    x_i(u)
      =
      \begin{cases}
        +s_i &\text{when $u = u_i$,} \\
        -s_i &\text{when $u = v_i$,} \\
        0 &\text{otherwise,}
      \end{cases}
  \]
  where $s_i \in \{ -1, +1 \}$ is chosen arbitrarily.
  Note that $L_0 = \frac{1}{m} \sum_{i=1}^m x_i x_i'$.
\item Take $s_i$ to be random, equal to $+1$ or $-1$ with probability
  $\tfrac{1}{2}$.  Approximate the distribution of $x_i$ by the
  distribution of a multivariate normal random variable, $\tilde x_i$,
  such that $x_i$ and $\tilde x_i$ have the same first and second
  moments.
\item Approximate the distribution of $L_0$ by the distribution of $\tilde L_0$, where
  \(
  \tilde L_0 = \frac{1}{m} \sum_{i=1}^m \tilde x_i \tilde x_i'.
  \)
  \item Use the asymptotic expansion above to approximate the
    distribution of $L$ by the distribution of
    $\Delta^{-1/2} \tilde L_0 \Delta^{-1/2}$.
\end{enumerate}

\noindent
The next two lemmas derive the distribution of $\tilde x_i$ and
$\tilde L_0$ in terms of $\mathcal{L}$, allowing us to get an
approximation for $p(L \mid \mathcal{L})$.

\begin{lemma}
  With $x_i$ and $\tilde x_i$ defined as above,
  \[
    \E_\omega[ x_i ] = \E_\omega[ \tilde x_i ] = 0,
  \]
  and
  \[
    \E_\omega[ x_i x_i' ] = \E_\omega [ \tilde x_i \tilde x_i' ] = \mathcal{L}_0.
  \]
\end{lemma}
\begin{proof}
The random variable $\tilde x_i$ is defined to have the same first
and second moments as $x_i$.
The first moment vanishes since $s_i \overset{d}{=} -s_i$ implies
that $x_i \overset{d}{=} -x_i$.  For the second moments, note that
when $u \neq v$, 
\[
  \E_\omega[x_i(u) \, x_i(v)]
  = -s_i^2 \, \prob_\omega\{ (u_i,v_i) = (u,v) \}  = -\omega(u,v)
  = \mathcal{L}_0(u,v).
\]
Likewise,
\[
  \E_\omega[\{x_i(u)\}^2]
      = \sum_{v} \prob_\omega\{ (u_i,v_i) = (u,v) \}
      = \sum_{v} \omega(u,v)
      = \mathcal{L}_0(u,u). \qedhere
\]
\end{proof}

\begin{lemma}\label{L:approx-wishart}
  The random matrix $\tilde L_0$ is distributed as $\tfrac{1}{m}
  \mathrm{Wishart }(\mathcal{L}_0, m)$ random variable.
  This distribution is supported on the set of
  positive-semidefinite matrices with the same nullspace as $\mathcal{L}_0$.  When
  $m \geq \rank(\mathcal{L}_0)$, the distribution has a density on this space
  given by
  \begin{equation}\label{E:wishart-density}
   f( \tilde L_0 \mid \mathcal{L}_0, m)
      \propto
      \frac{|\tilde L_0|^{(m - \rank(\mathcal{L}) - 1)/2}
        \exp\{-\tfrac{m}{2} \Tr(\tilde L_0 \mathcal{L}_0^+) \}}
        {|\mathcal{L}_0|^{m/2}}
  \end{equation}
  where the constant of proportionality depends only on $m$ and $n$
  and where $|\cdot|$ denotes pseudodeterminant (product of nonzero
  eigenvalues).
\end{lemma}
\begin{proof}
  Since $m \tilde L$ is a sum of $m$ outer products of multivariate
  $\mathrm{Normal}(0, \mathcal{L}_0)$, it is Wishart distributed
  (by definition).
  Suppose $\rank(\mathcal{L}_0) = r$ and
  $U \in \reals^{n \times r}$ is a matrix whose columns are the
    eigenvectors of $\mathcal{L}_0$.  Note that
    $U' \tilde x_i \overset{d}{=} \mathrm{Normal}(0, U' \mathcal{L}_0 U)$,
    and that $U' \mathcal{L}_0 U$ has full rank.  Thus,
    \(
      U' \tilde L_0 U
    \)
    has a density over the space of $r \times r$ positive-semidefinite
    matrices whenever $m \geq r$.  The density of $U' \tilde L U$ is
    exactly equal to $f(\tilde L_0 \mid \mathcal{L}_0, m)$,
    defined above.
\end{proof}

Using the previous lemma, the random variable $\tilde L = \Delta^{-1/2} \tilde L_0
\Delta^{-1/2}$ has density
\[
  f(\tilde L \mid \mathcal{L}, m)
    \propto
     \frac{|\Delta^{1/2} \tilde L \Delta^{1/2}|^{(m - \rank(\mathcal{L}) - 1)/2}
        \exp\{-\tfrac{m}{2} \Tr(\Delta^{1/2} \tilde L \Delta^{1/2} \mathcal{L}_0^+) \}}
        {|\Delta^{1/2} \mathcal{L}_0 \Delta^{1/2}|^{m/2}},
\]
where we have used that $\rank(\mathcal{L}_0) = \rank(\mathcal{L})$, and
the constant of proportionality depends on $m$, $n$,
$\rank(\mathcal{L})$, and $\Delta$.  
Then, if we approximate 
$| \Delta^{1/2} \tilde L \Delta^{1/2}| \approx |\Delta| |\tilde L|$ and
$\Delta^{1/2} \mathcal{L}_0^+ \Delta^{1/2} \approx \mathcal{L}^+$,
then $f$ is ``approximately''  the density of a $\tfrac{1}{m}
\mathrm{Wishart }(\mathcal{L}, m)$ random variable.  These last
approximations are necessary because $\tilde L$ and $\mathcal{L}_0$
are rank-degenerate.

To conclude, we do not want to overstate the validity of this heuristic 
justification.
In particular, it makes three key approximations:
\begin{enumerate}
\item the true degree matrix $\Delta$ can be
  approximated by the observed degree matrix $D$;
\item the distribution of $x_i$, a sparse vector, is well approximated
  $\tilde x_i$, a Gaussian (dense) vector;
\item the quantities
  $| \Delta^{1/2} \tilde L \Delta^{1/2}|$
  and
  $\Delta^{1/2} \mathcal{L}_0^+ \Delta^{1/2}$
  can be replaced with
  $|\Delta| |\tilde L|$
  and $\mathcal{L}^+$.
\end{enumerate}
None of these approximations hold in general, though as argued above, the 
first is plausible if $m$ is large relative to $n$.  
Likewise, since $\tilde L$ and $\mathcal{L}$ are nearly full rank, the third
approximation is likely not too bad.  
The biggest leap of faith is the second approximation.  
Note, \emph{e.g.}, that despite their first moments being equal, the 
second moments of $\tilde x_i \tilde x_i'$ and $x_i x_i'$ differ.

\vspace{-2mm}
\section{Other priors and the relationship to Heat Kernel and Lazy Random Walk}
\label{sxn:other-priors}
\vspace{-2mm}

There is a straightforward generalization of Proposition~\ref{P:map-sdp} to 
other priors.
In this section, we state it, and we observe connections with the Heat 
Kernel and Lazy Random Walk procedures.

\begin{proposition}
\label{prop:map-generaliz}
  Suppose the conditional likelihood for $L$ given $\mathcal{L}$ is as
  defined in \eqref{E:density} and the prior density for $\mathcal{L}$
  is of the form
  \begin{equation}
    p(\mathcal{L})
      \propto
        p(\tau)
        |\Theta|^{-m/2}
        \exp\{ -q(\tau) \, G(\Theta) \},
  \label{eqn:prior-app}
  \end{equation}
  where
  $\tau = \Tr(\mathcal{L}^{+})$,
  $\Theta = \tau^{-1} \mathcal{L}^{+}$,
  and $p$ and $q$ are functions with $q(\tau) > 0$ over the support
  of the prior.
  Define
  $\mathcal{\hat L}$ to be the MAP estimate of $\mathcal{L}$.  Then,
  $[\Tr(\mathcal{\hat L}^+)]^{-1} \mathcal{\hat L}^+$ solves the
  Mahoney-Orecchia regularized SDP \eqref{eqn:mo-reg-sdp}, with $G$ the same 
  as in the expression (\ref{eqn:prior-app}) for $p(\mathcal{L})$ and with
  \[
    \eta = \frac{m \hat \tau}{2 \, q(\hat \tau)},
  \]
  where $\hat \tau = \Tr(\mathcal{\hat L}^+)$.
\end{proposition}

The proof of this proposition is a straightforward generalization of the 
proof of Proposition~\ref{P:map-sdp} and is thus omitted.  
Note that we recover the result of Proposition~\ref{P:map-sdp} by setting
$G(\Theta) = - \log |\Theta|$ and $q(\tau) = \frac{m}{2} + \alpha - 1$.
In addition, by choosing $G(\cdot)$ to be the generalized entropy or the matrix 
$p$-norm penalty of~\cite{MO11-implementing}, we obtain variants of the 
Mahoney-Orecchia regularized SDP \eqref{eqn:mo-reg-sdp} with the 
regularization term $G(\cdot)$.
By then combining Proposition~\ref{prop:map-generaliz} with their result, we 
get that the MAP estimate of $\mathcal{L}$ is related to the Heat Kernel and 
Lazy Random Walk procedures, respectively, in a manner analogous to what we 
saw in Section~\ref{sxn:priors} with the PageRank procedure.
In both of these other cases, however, the prior $p(\mathcal{L})$ is 
data-dependent in the strong sense that it explicitly depends on the number 
of data points; and, in addition, the priors for these other cases do not 
correspond to any well-recognizable parametric distribution.

\vspace{-2mm}
\section{Regularization performance with respect to the relative spectral error}
\label{S:spectral-performance}
\vspace{-2mm}

In this section, we present Figure~\ref{fig:perf-spec}, which shows 
the regularization performance for our empirical evaluation, when the 
performance criterion is the relative spectral norm error, \emph{i.e.},
$\|\Theta - \hat \Theta_\eta\|_\mathrm{2} / \|\Theta - \hat \Theta
\|_\mathrm{2}$, where $\|\cdot\|_\mathrm{2}$ denotes spectral norm of a 
matrix (which is the largest singular value of that matrix).
See Section~\ref{S:estimation} for details of the setup.
Note that these results are very similar to those for the relative Frobenius 
norm error that are presented in Figure~\ref{fig:perf-frob}.

\begin{figure}[h]
  \centering
  \subfigure[$m/\mu = 0.2$ and $s = 0$.]{
    \makebox{\includegraphics[scale=0.34]{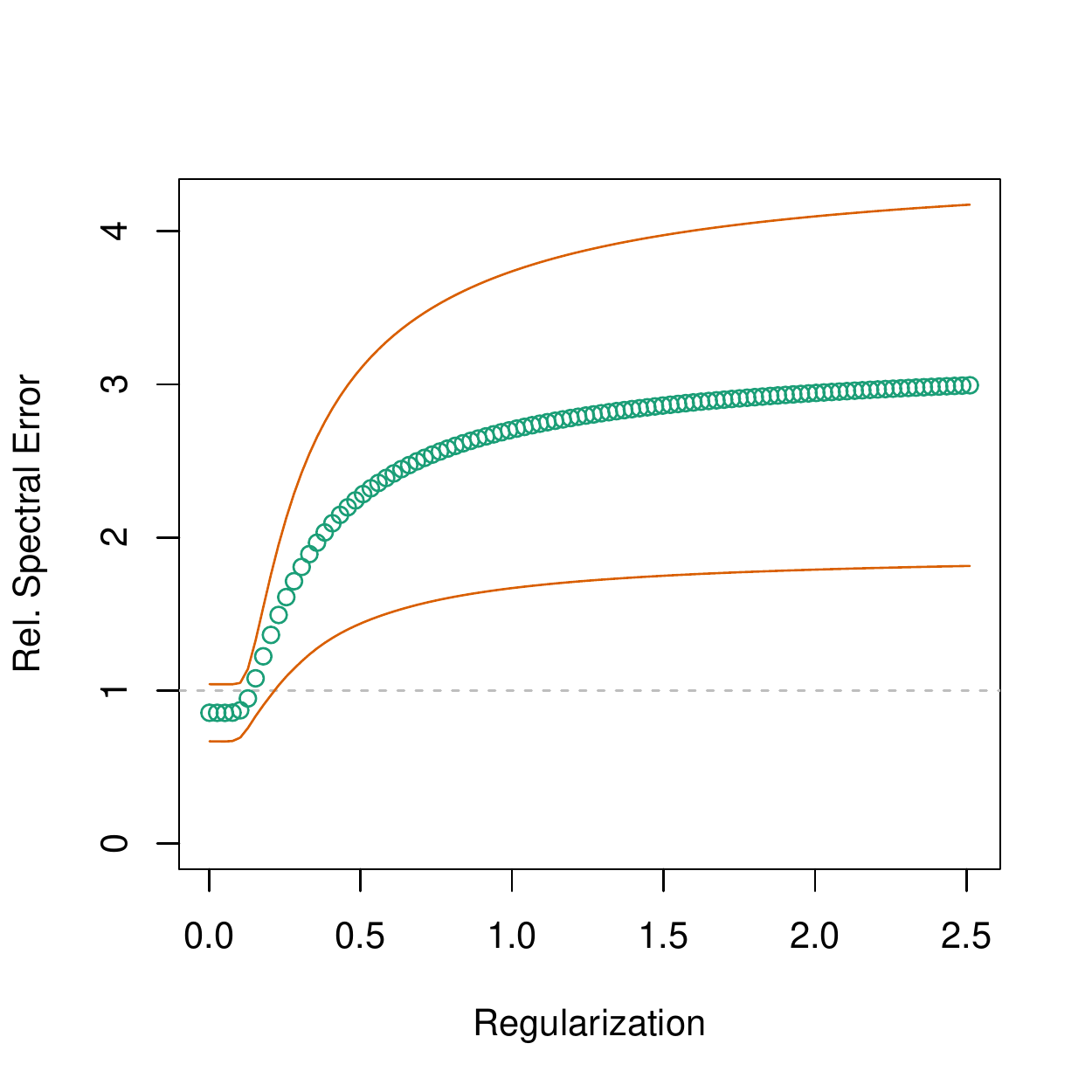}}
  }
  \subfigure[$m/\mu = 0.2$ and $s = 4$.]{
    \makebox{\includegraphics[scale=0.34]{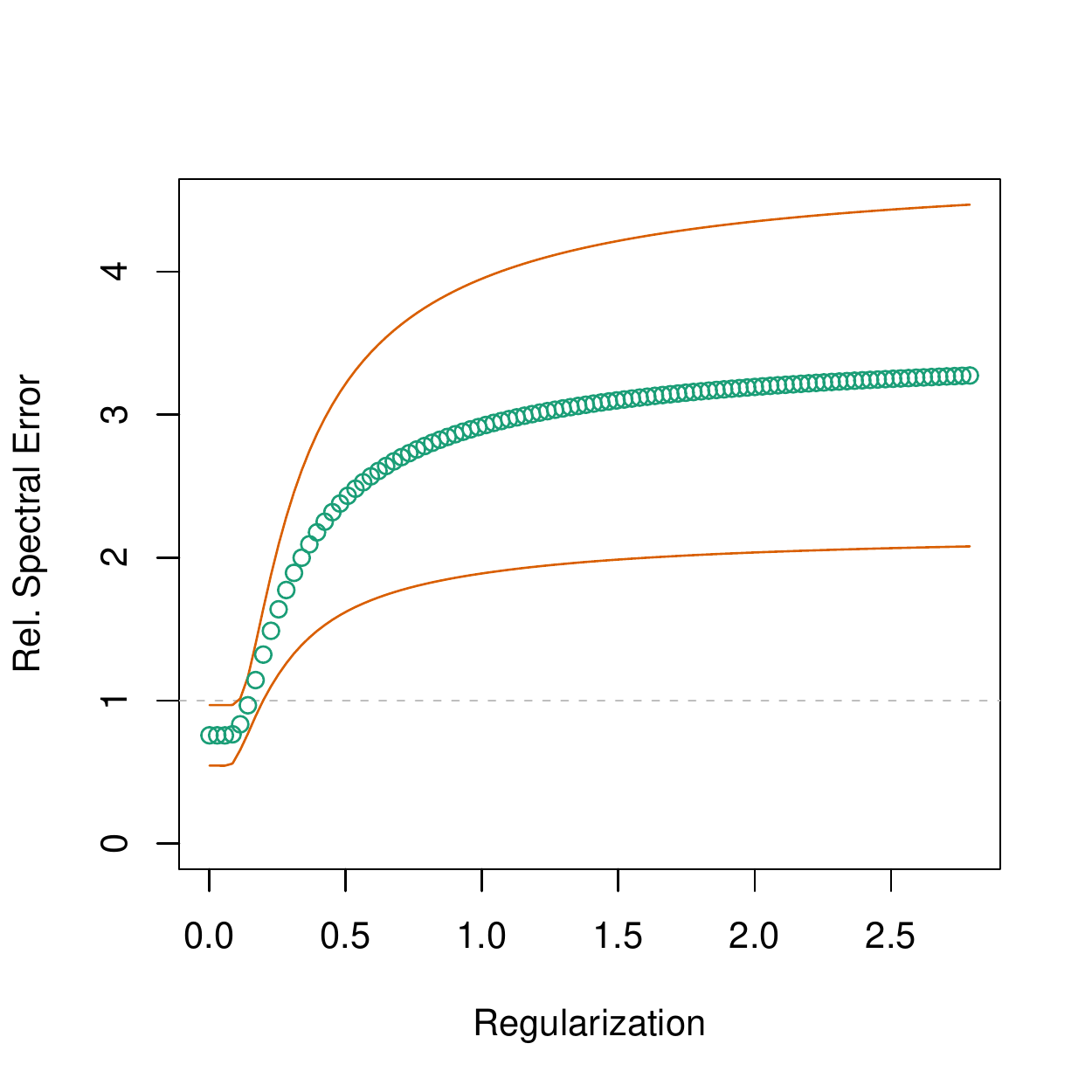}}
  }
  \subfigure[$m/\mu = 0.2$ and $s = 32$.]{
    \makebox{\includegraphics[scale=0.34]{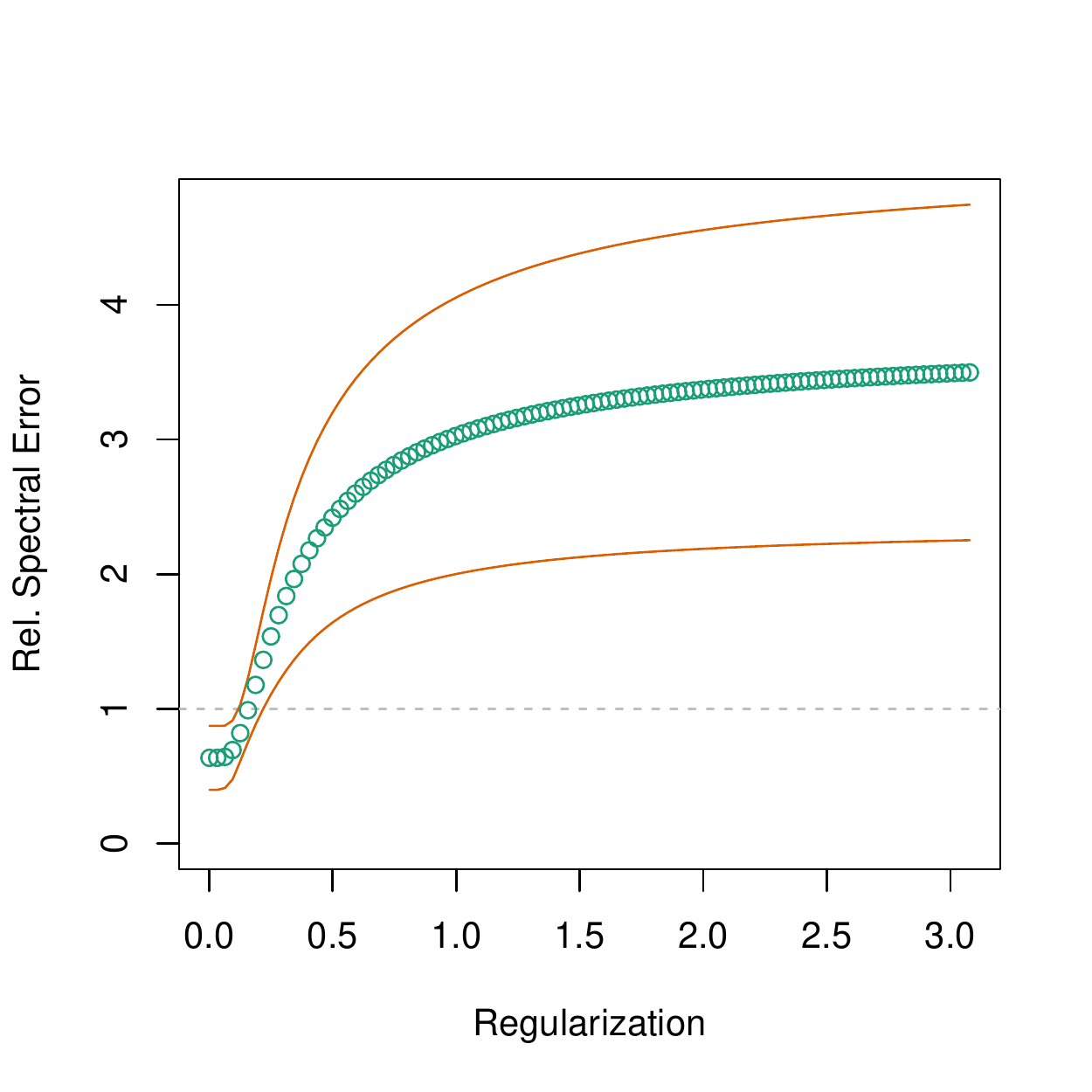}}
  } \\
  \subfigure[$m/\mu = 2.0$ and $s = 0$.]{
    \makebox{\includegraphics[scale=0.34]{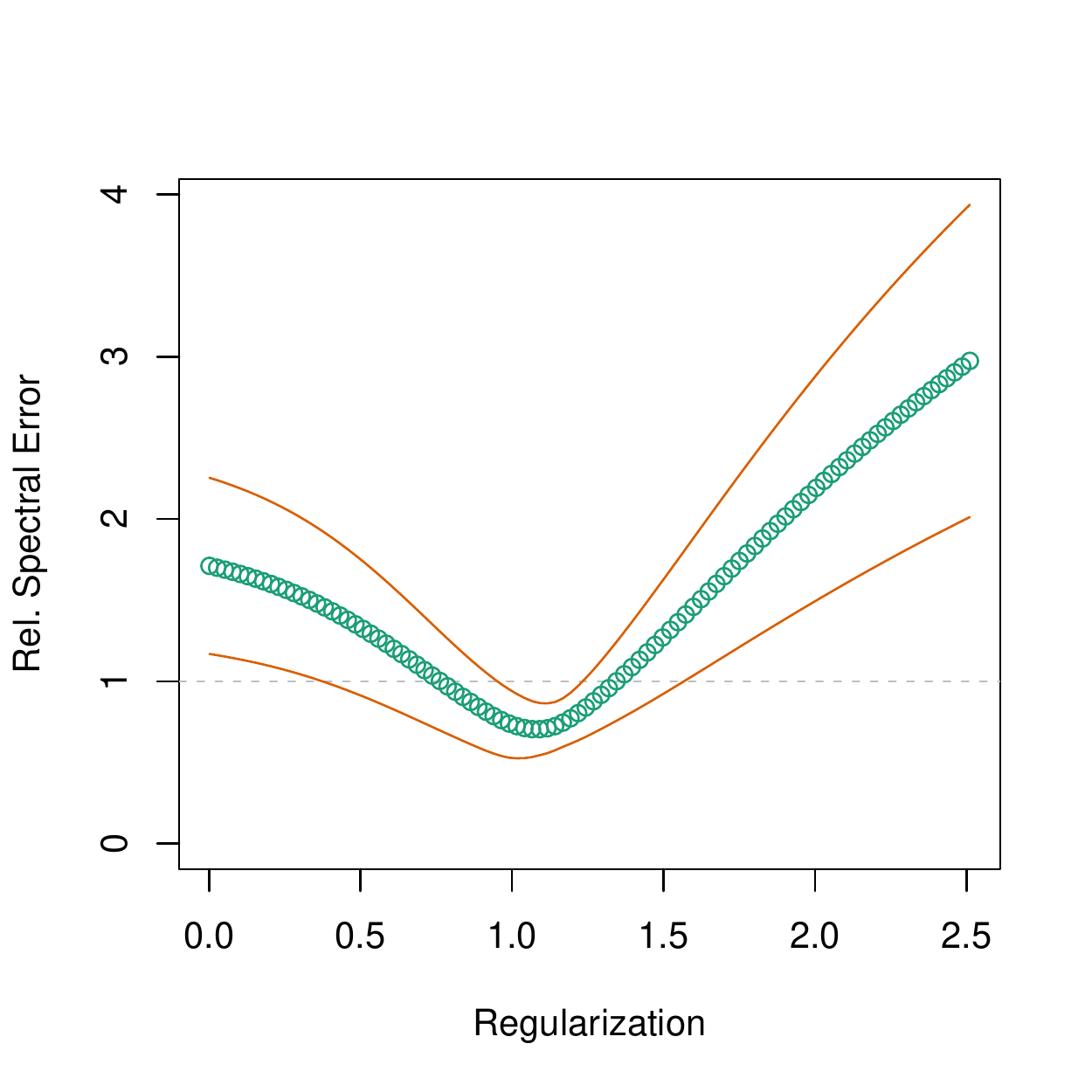}}
  }
  \subfigure[$m/\mu = 2.0$ and $s = 4$.]{
    \makebox{\includegraphics[scale=0.34]{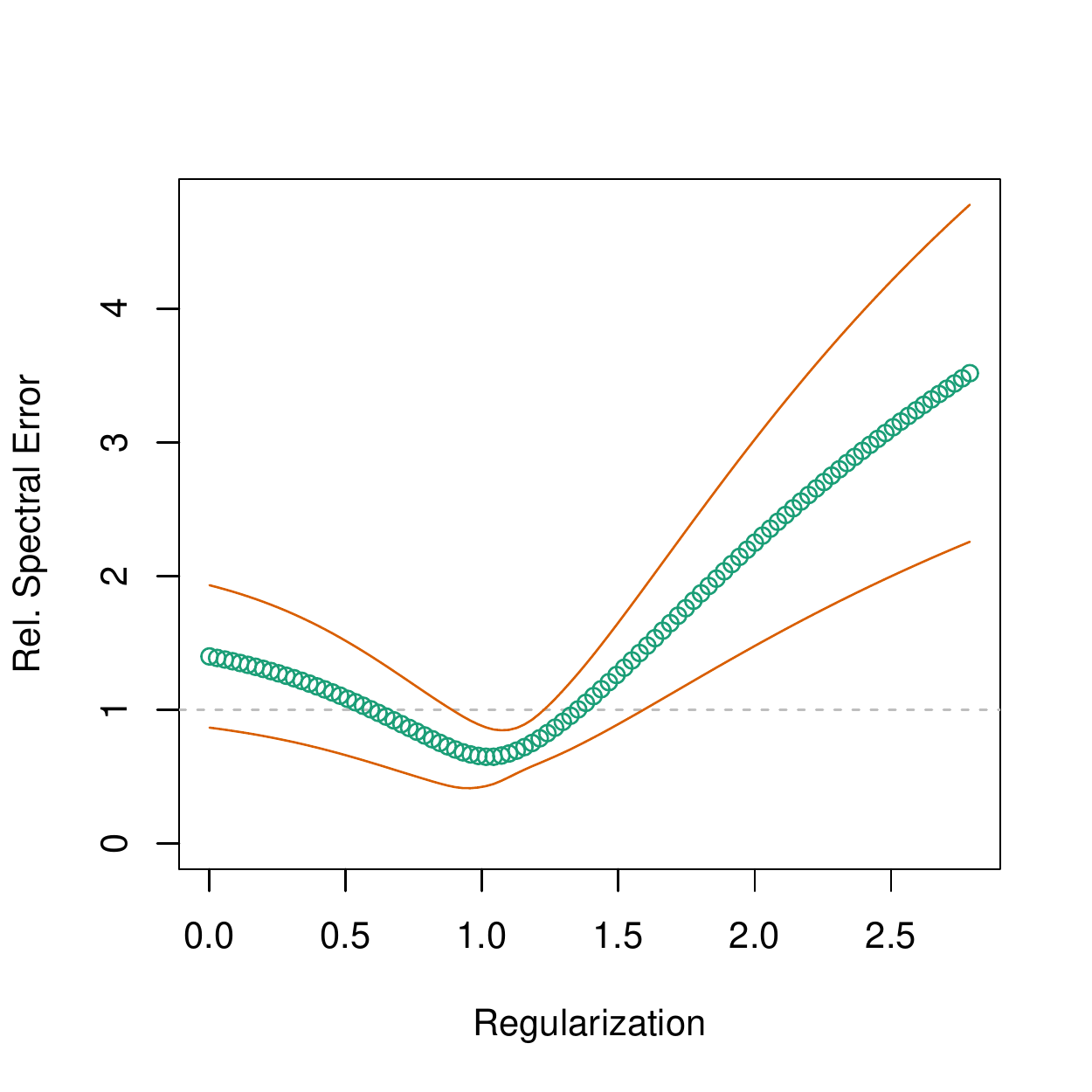}}
  }
  \subfigure[$m/\mu = 2.0$ and $s = 32$.]{
    \makebox{\includegraphics[scale=0.34]{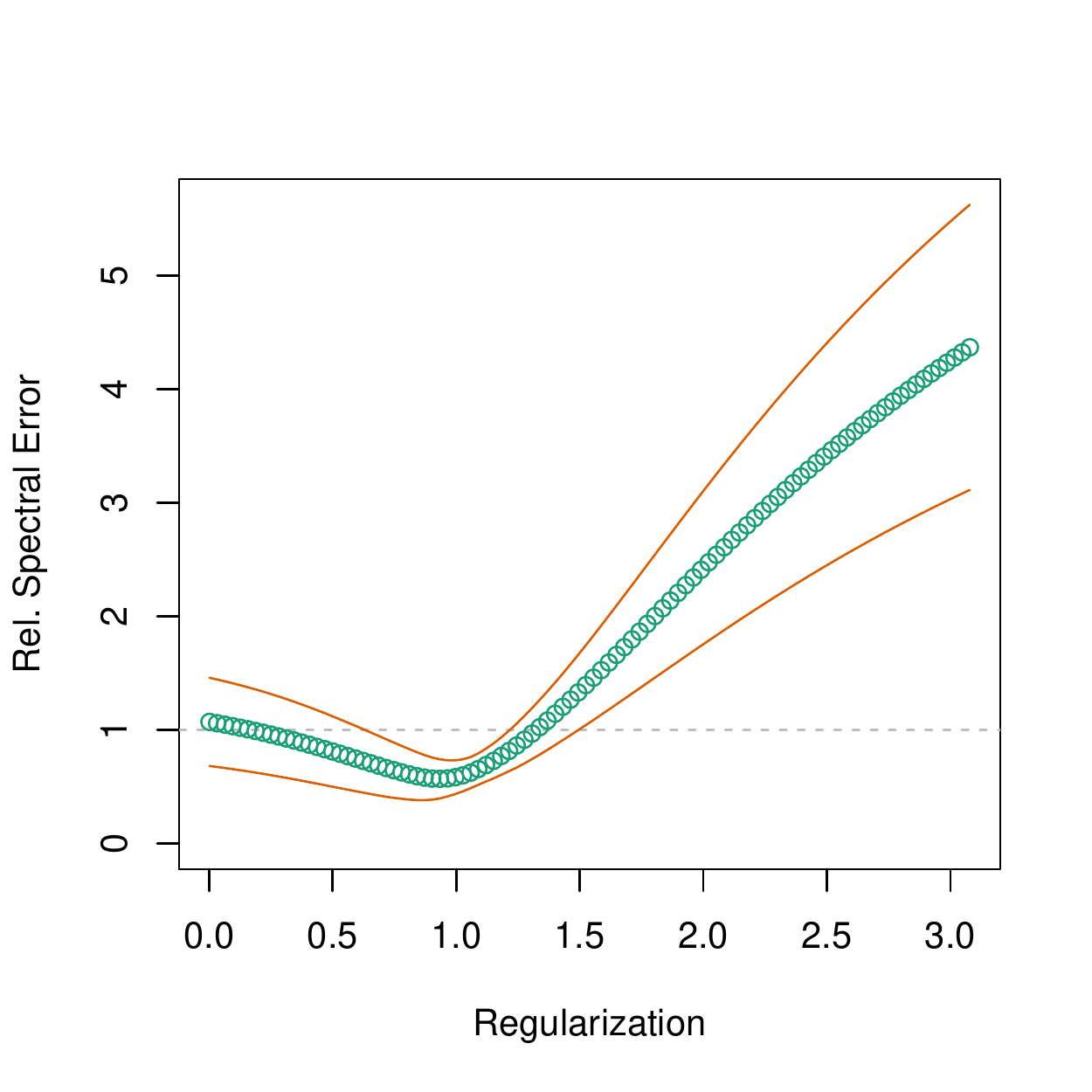}}
  }
  \vspace{-1em}
\caption{Regularization performance, as measured with the relative 
spectral norm error, versus the (normalized) regularization parameter 
$\eta / \bar \tau$.
Shown are plots for various values of the (normalized) number of edges, 
$m/\mu$, and the edge-swap parameter, $s$.
Recall that the regularization parameter in the regularized 
SDP~(\ref{eqn:mo-reg-sdp}) is $1/\eta$, and thus smaller values along the 
X-axis correspond to stronger regularization.}
\label{fig:perf-spec}
\end{figure}

\end{document}